\documentclass[11pt, a4paper]{article}

\usepackage[english]{babel}
\usepackage[autostyle]{csquotes}
\usepackage[T1]{fontenc}
\usepackage{microtype}
\usepackage[shortlabels]{enumitem}

\usepackage{lmodern}
\usepackage{geometry}

\usepackage{customoperators}
\usepackage{bibliographysettings}

\usepackage[pdftex, bookmarksnumbered,colorlinks=false]{hyperref}
\usepackage{bookmark}

\theoremstyle{plain}
\newtheorem{thm}{Theorem}[]

\newtheorem{lem}[thm]{Lemma}
\newtheorem{cor}[thm]{Corollary}
\newtheorem{assu}{Assumption}[]

\theoremstyle{definition}

\newtheorem*{rem}{Remark}
\newtheorem*{rems}{Remarks}

\makeatletter
\def\dicho#1{\expandafter\@dicho\csname c@#1\endcsname}
\def\@dicho#1{\ifnum#1>1 or \else\fi(\@roman#1)}
\makeatother
\AddEnumerateCounter{\dicho}{\@dicho}{or (iii)}
\newlist{dichotomy}{enumerate}{1}
\setlist[dichotomy]{label=\dicho*,ref=(\roman*),leftmargin=1.5cm}

\ifpdf
\hypersetup{
  pdfauthor={Lukas Reichmann},
  pdftitle={Dynamics of two-dimensional Bose gases in the focusing case},
}
\fi

\title{Dynamics of two-dimensional Bose gases in the focusing case}
\author{Lukas Reichmann\footnote{Department of Mathematics, Friedrich-Alexander-Universit{\"a}t Erlangen-N{\"u}rnberg, Cauerstrasse 11, 91058 Erlangen, Germany}}
\date{\today}

\begin{document}
\maketitle


\pdfbookmark[section]{Abstract}{abstract} 
\begin{abstract}
We prove the propagation of Bose--Einstein condensation and a norm approximation for the dynamics of a Bose gas with localized initial data in two dimensions with singular interactions that may have a small attractive part. We allow the scaling of the interaction to be any polynomial scaling $N^\beta$ for $0<\beta<\infty$ and even exponential scaling $\ee^{N^\kappa}$ for $0<\kappa<1/4$. Our approach covers external potentials that are time-dependent, at most quadratic, and bounded below, including the situation without an external potential. 
\end{abstract}

\section{Introduction}
Bose--Einstein condensates were already predicted by Bose and Einstein \cite{Bose-24,Ein-24,Ein-25} in 1924, while the first experimental realizations of Bose--Einstein condensates were achieved in 1995 \cite{Cor-95,Ket-95}. To describe the time evolution of such condensates, the Hartree and nonlinear Schr{\"o}dinger (NLS) equations are used, which were first rigorously derived in \cite{Hepp-74,GinVel-79,Spo-80}. In the repulsive case (i.e., the interaction potential is non-negative, $w\geq0$), the first rigorous results justifying the Hartree and NLS equations in the dilute regime (i.e., the interaction potential depends on the number of particles $N$, $w_N=N^{d\beta}w(N^\beta\cdot)$, for $\beta>0$ and dimension $d$) were obtained in \cite{AdaGolTet-07} for dimension $1$, in \cite{ErdSchYau-07,ErdSchYau-09,ErdSchYau-10} for dimension $3$, and in \cite{KirSchSta-11} for dimension $2$ (see \cite{JebLeoPic-19} for a later result, which additionally covers the Gross--Pitaevskii regime). 

The attractive case ($w\leq0$) was first dealt with in \cite{CheHol-16} for dimension $1$ and in \cite{CheHol-17,JebPic-18} for dimension $2$. Both approaches \cite{CheHol-17,JebPic-18} rely on stability of second kind provided by an external potential and prove the convergence of the reduced density matrix associated with the dynamics. In \cite{NamNap-19,BosDieNam-26} these results were extended to Hamiltonians without an external potential by proving norm convergence of the dynamics, which implies the convergence of the reduced density matrices. The result in \cite{NamNap-19} covers all $0<\beta<\infty$ in dimension $1$ and $0<\beta<1$ in dimension $2$ by using a weaker form of stability as a substitute for the stability of second kind. The approach in \cite{BosDieNam-26} covers $0<\beta<3/2$ in dimension $2$ and even includes the instability case, where the solution of the NLS blows up after a finite time. For results on the derivation of effectively 2D dynamics of strongly confined 3D systems, we refer to \cite{CheHol-13,Bos-20}. 

In the present work, we want to prove a norm approximation in dimension $2$ following the approach in \cite{NamNap-19,BosDieNam-26}. Results on the ground state of trapped systems (see e.g.\ \cite{BreSchSch-22,GreSei-13,LewNamSerSol-15,NamSei-15,NamTri-23}) suggest it is reasonable to assume that the initial data satisfy Bose--Einstein condensation (BEC) and the kinetic energy of the particles orthogonal to the condensate is bounded. These results are all in the mean-field or defocusing case. The closest result in the 2D focusing case is \cite{JunVis-26}, where it is shown that the leading order of the ground state energy of a trapped system with singular and possibly attractive interactions is described by the ground state energy of the corresponding NLS functional and ground states satisfy Bose--Einstein condensation. 

If the negative part of the interaction is small enough, we have stability of second kind (see \cite[Chapter~3]{LieSei-09} for a definition) of the trapped Hamiltonian \cite{LewNamRou-16,LewNamRou-17,NamRou-20,JunVis-26}, i.e., the Hamiltonian $H_N^\textnormal{trap}$ satisfies the lower bound $H_N^\textnormal{trap}\geq -CN$. Hence, the ground state energy of the trapped system is not too negative. 
This has been used in \cite{CheHol-17,JebPic-18} to bound the kinetic energy of the dynamics, which in turn is used to bound the negative interaction potential. 

Bogoliubov's approximation \cite{Bog-47} suggests that the time evolution of the particles orthogonal to the condensate is approximately described by an effective quadratic (in creation and annihilation operators) Hamiltonian in Fock space. In dimension $3$, this has been rigorously justified in \cite{LewNamSch-15, MitPetPic-19,BosPetPicSof-21} in the mean-field regime $\beta=0$, in \cite{NamNap-17a,NamNap-17b,BreNamNapSch-19} in the dilute regime $\beta\in(0,1)$, and in \cite{CarOldSch-24} in the Gross--Pitaevskii regime $\beta=1$. 
For norm approximations around coherent states, see e.g.\ \cite{GriMacMar-10,GriMacMar-11} in the mean-field regime and \cite{BocCenSch-17,GriMac-13,GriMac-17,Kuz-17} in the dilute regime. 
All these results are either in the mean-field regime or require the interaction potential to be non-negative; for results in three dimensions in the direction of the focusing case, see \cite{Cho-21,NamNap-17a,Pic-10}. 

Norm convergence in two dimensions in the focusing case has previously been achieved in \cite{NamNap-19} for $0<\beta<1$ under the stability assumption
\begin{equation}\label{eq:stabilityassumption}
	\int_{\R^2}w_-<a^\ast,
\end{equation}
where $w_-\coloneqq\max(-w,0)$ is the negative part of the interaction and $a^\ast$ is the optimal constant in the Gagliardo--Nirenberg inequality, and in \cite{BosDieNam-26} for $0<\beta<3/2$ without any assumptions on stability. 
We follow the strategy used in \cite{NamNap-19,BosDieNam-26} to prove norm convergence under the stability assumption \eqref{eq:stabilityassumption} for all $0<\beta<\infty$ and even for exponential scaling $w_N=\ee^{2N^\kappa}w(\ee^{N^\kappa}\cdot)$ for $0<\kappa<1/4$, which corresponds to $N$-dependent $\beta$ of the form $\beta=N^\kappa/\log(N)$. The key ingredient to extend the range of $\beta$ is a new bound on the kinetic energy of the dynamics for localized initial data, which we prove using \cite[Theorem~2]{JunVis-26}. 
Our result extends results from \cite{CheHol-17,JebPic-18} on convergence of density matrices to the NLS solution and from \cite{NamNap-19,BosDieNam-26} on norm convergence of the dynamics under the stability assumption \eqref{eq:stabilityassumption}. This answers a question raised in \cite[Remark~2]{JunVis-26}. 

\subsection{Setting}

We consider $N$ identical (spinless) bosons in $\R^2$. Such a system is described by a wave function in the bosonic Hilbert space $L^2_\textnormal{s}(\R^{2N})$, the space of square-integrable permutation-symmetric functions of $N$ variables $x_j\in\R^2$. The evolution of this system is governed by the linear $N$-body Schr{\"o}dinger equation
\begin{equation}\label{eq:linearschr}
\begin{cases}
	\ii\partial_t\Psi_N(t)=H_{N,t}\Psi_N(t), \\
	\Psi_N(0)=\Psi_{N,0},
\end{cases}
\end{equation}
where the (time-dependent) Hamiltonian $H_{N,t}$ is given by
\begin{equation}\label{eq:Hamiltonian}
	H_{N,t}=\sum_{j=1}^N\bigl(-\Delta_j+V(t,x_j)\bigr)+\frac{1}{N-1}\sum_{1\leq j<k\leq N}w_N(x_j-x_k).
\end{equation}
The interaction potential $w_N$ should satisfy the following assumptions:
\begin{assu}\label{assu:wN}
Either 
\begin{equation}\label{eq:poly}
	w_N(x)=N^{2\beta}w(N^\beta x)
\end{equation}
for any $0<\beta<\infty$ or 
\begin{equation}\label{eq:exp}
	w_N(x)=\ee^{2N^\kappa}w(\ee^{N^\kappa}x)
\end{equation}
for $0<\kappa<1/4$. 
The fixed function $w\colon\R^2\to\R$ satisfies
\begin{equation}\label{assu:w}
	w\in L^1(\R^2)\cap L^2(\R^2),\qquad \int_{\R^2}w_-<a^\ast,\qquad w(x)=w(-x),\qquad w^2\leq C(1-\Delta)
\end{equation}
for some constant $C>0$, where $w_-\coloneqq\max(-w,0)$ and $a^\ast$ is the optimal constant in the Gagliardo--Nirenberg inequality
\begin{equation}\label{eq:GN}
	\biggl(\int_{\R^2}\abs{\nabla f}^2\biggr)\biggl(\int_{\R^2}\abs{f}^2\biggr)\geq\frac{a^\ast}{2}\int_{\R^2}\abs{f}^4,\qquad \forall f\in H^1(\R^2).
\end{equation}
\end{assu}

For references on the Gagliardo--Nirenberg inequality, see e.g.\ \cite{GuoSei-14,Wei-83,Mae-10,Zha-00}. 
We make the following assumptions on the external potential $V$:
\begin{assu}\label{assu:V}
\begin{enumerate}
\item\label{assu:Va}
$V\in L^\infty_\textnormal{loc}(\R\times\R^2)$ is real-valued. 
\item\label{assu:Vb}
$V(t,\cdot)\in C^\infty(\R^2)$ for all $t\in\R$.
\item\label{assu:Vc}
$V$ is at most quadratic in space, uniformly with respect to time, in the sense that
\[
	\forall\bm{\alpha}\in\N_0^2,\;\abs{\bm{\alpha}}\geq 2:\quad\partial^{\bm{\alpha}}_xV\coloneqq \partial_{x_1}^{\alpha_1}\partial_{x_2}^{\alpha_2}V\in L^\infty(\R\times\R^2).
\]
Additionally, $V(\cdot,x),\nabla V(\cdot,x)\in L^\infty(\R)$ for all $x\in\R^2$. These assumptions imply $\abs{V(t,x)}\leq C(1+\abs{x}^2)$ and $\abs{\nabla V(t,x)}\leq C\sqrt{1+\abs{x}^2}$ for all $(t,x)\in\R\times\R^2$ and some constant $C>0$. 
\item\label{assu:Vd}
There exists some $0\leq v\in L^\infty_\textnormal{loc}(\R^2)$ such that $C_t^{-1}v(x)-c_t\leq V(t,x)\leq C_tv(x)+c_t$ for some (continuously) time-dependent constants $C_t,c_t>0$.
\item\label{assu:Ve}
$V(\cdot,x)\in C^1(\R)$ for all $x\in\R^2$ and $\abs{\partial_t V(t,x)}\leq C_t(v(x)+1)$, where $v$ is defined in \ref{assu:Vd} and $C_t>0$ is some (continuously) time-dependent constant.
\end{enumerate}
\end{assu}

\begin{rems}
\begin{enumerate}
\item
For example, a (time-dependent) harmonic potential of the form
\[
	V(t,x) = \langle x,Q(t)x\rangle
\]
satisfies Assumption~\ref{assu:V}, where $Q(t)\in\R^{2\times 2}$ is a positive-definite matrix, $C\geq Q(t)\geq C_t^{-1}>0$, with entries that are continuously differentiable in $t$, or $Q\equiv 0$, which corresponds to $V\equiv0$. Such time-dependent harmonic potentials are typically used to model trapped Bose--Einstein condensates; see e.g.\ \cite{CasDum-96,GRPGV-01,RKRadPan-08}.
\item
The condition $w^2\leq C(1-\Delta)$ and the assumptions \ref{assu:Vc} and \ref{assu:Vd} on $V$ guarantee that the Hamiltonian $H_{N,t}$ can be defined as a self-adjoint operator on the same domain for every $t$. 
Using assumption \ref{assu:Ve} on $V$, this allows the evolution $\Psi_N$ to be properly defined as the global solution to the linear Schr{\"o}dinger equation \eqref{eq:linearschr} by applying \cite[Theorem~X.70]{ReeSim-2}. 
\item
We can think of the initial state $\Psi_{N,0}$ as the ground state of a trapped Hamiltonian
\[
	H_N^W=H_{N,0}+\sum_{j=1}^NW(x_j),
\]
where $W(x)+V(0,x)\geq C\abs{x}^s - c$ for some constants $s,C,c>0$, which ensures that the Hamiltonian $H_N^W$ is stable of second kind; see \cite{JunVis-26}. After the trap $W$ is turned off, the evolution of $\Psi_{N,0}$ is described by $\Psi_N(t)$ in \eqref{eq:linearschr}. Our proof requires the Hamiltonian with a suitable auxiliary external potential to be stable of second kind, which is why we need the assumption 
\[
	\int_{\R^2}w_- < a^\ast
\]
on the interaction potential. This stability will later be used to bound the kinetic energy of $\Psi_N$. 
\end{enumerate}
\end{rems}

\subsection{Notation}

$C$ always denotes a positive constant that does not depend on time and may vary from line to line. We will denote by $C_t$ a positive constant that is continuous with respect to time and may also vary from line to line. We will sometimes omit the integration variables if there is no ambiguity. We will use the notation $\jbr{x}\coloneqq (1+\abs{x}^2)^{1/2}$ and $V_t\coloneqq V(t,\cdot)$. When it is unambiguous, we will abbreviate the $L^2$-norm as well as the norm in Fock space and their corresponding inner products by $\norm{\cdot}$ and $\langle\cdot,\cdot\rangle$, respectively. The one-particle reduced density matrix of a wave function $\Psi\in L^2_\textnormal{s}(\R^{2N})$ will be denoted by
\[
	\gamma_{\Psi}^{(1)} = \Tr_{2,\dotsc,N}\ketbra{\Psi}{\Psi}. 
\]
We will also use the notation
\begin{equation*} 
	\Sigma^k\coloneqq\biggl\{f\in L^2(\R^2):\norm{f}_{\Sigma^k}\coloneqq\sum_{\abs{\bm{\alpha}}+\abs{\bm{\beta}}\leq k}\norm{x^{\bm{\alpha}}\partial^{\bm{\beta}}_x f}_{L^2}<\infty\biggr\}
\end{equation*}
for $k\in\N$ and $\bm{\alpha},\bm{\beta}\in\N_0^2$, where $x^{\bm{\alpha}}\coloneqq x_1^{\alpha_1}x_2^{\alpha_2}$ and $\partial^{\bm{\beta}}_x\coloneqq\partial_{x_1}^{\beta_1}\partial_{x_2}^{\beta_2}$. 

\section{Main results and proof strategy}

\subsection{Main results}

We assume that the initial state $\Psi_{N,0}$ in \eqref{eq:linearschr} exhibits BEC, that is, $\Psi_{N,0}\approx \varphi_0^{\otimes N}$ in an appropriate sense. It is expected that the Schr{\"o}dinger flow preserves this property. To describe the evolution of the condensate $\varphi_0$, we formally substitute the ansatz $\Psi_N(t)\approx u_N(t)^{\otimes N}$ in \eqref{eq:linearschr}, which yields the Hartree equation
\begin{equation}\label{eq:Hartree}
\begin{cases}
	\ii\partial_tu_N(t,x)=\bigl(-\Delta_x+V(t,x)+(w_N\ast\abs{u_N(t,\cdot)}^2)(x)-\mu_N(t)\bigr)u_N(t,x), \\
	u_N(0,x)=\varphi_0(x),
\end{cases}
\end{equation}
where 
\[
	\mu_N(t)=\frac{1}{2}\int_{\R^2}\int_{\R^2}\abs{u_N(t,x)}^2w_N(x-y)\abs{u_N(t,y)}^2\diff x\diff y.
\]
Let us define
\begin{equation}\label{eq:Hartree:h}
	h(t)=-\Delta_x+V(t,x)+(w_N\ast\abs{u_N(t,\cdot)}^2)(x)-\mu_N(t).
\end{equation}
In the limit $N\to\infty$, the Hartree evolution $u_N$ can be approximated by the solution $\varphi$ of the nonlinear Schr{\"o}dinger equation
\begin{equation}\label{eq:NLS}
\begin{cases}
	\ii\partial_t\varphi(t,x)=\bigl(-\Delta_x+V(t,x)+a\abs{\varphi(t,x)}^2-\mu(t)\bigr)\varphi(t,x), \\
	\varphi(0,x)=\varphi_0(x),
\end{cases}
\end{equation}
where 
\[
	a=\int_{\R^2}w,\qquad \mu(t)=\frac{a}{2}\int_{\R^2}\abs{\varphi(t,x)}^4\diff x.
\]
Our main results, Theorem~\ref{thm:norm} and Corollary~\ref{cor:condensate}, show that the time evolution of the condensate is described by equations~\eqref{eq:Hartree} and \eqref{eq:NLS}. 

In Theorem~\ref{thm:norm}, we also prove a description of the excitations orthogonal to the condensate using Bogoliubov theory. For this purpose, we use the Fock space formalism. We denote the bosonic Fock space by
\[
	\mathcal{F}=\mathcal{F}(\mathfrak{H})=\bigoplus_{k=0}^\infty\mathfrak{H}^k, 
	\qquad \mathfrak{H}^k=\bigotimes_\textnormal{sym}^k\mathfrak{H},
	\qquad \mathfrak{H}=L^2(\R^2). 
\]
Similarly, we define the excited Fock space 
\[
	\mathcal{F}_\perp(t)=\mathcal{F}(\mathfrak{H}_\perp)=\bigoplus_{k=0}^\infty\mathfrak{H}_\perp^k
\]
and the subspace 
\[
	\mathcal{F}_\perp^{\leq N}(t)=\bigoplus_{k=0}^N\mathfrak{H}_\perp^k,
\]
where the one-body excited space $\mathfrak{H}_\perp$ is given by
\[
	\mathfrak{H}_\perp=\{u_N(t)\}^\perp\subset \mathfrak{H}.
\]
Here, $u_N$ denotes the solution to the Hartree equation \eqref{eq:Hartree}. 
Note that there is a natural embedding $\mathcal{F}_\perp^{\leq N}(t)\subset\mathcal{F}_\perp(t)\subset\mathcal{F}$. 
For $g\in \mathfrak{H}$ we denote the standard creation and annihilation operators by $a^\ast(g)$ and $a(g)$, respectively. Additionally, we will use the operator-valued distributions $a^\ast_x$ and $a_x$ with $x\in\R^2$ such that
\[
	a^\ast(g)=\int_{\R^2} g(x)a^\ast_x\diff x,\qquad a(g)=\int_{\R^2} \overline{g(x)}a_x\diff x, \qquad \forall g\in\mathfrak{H}. 
\]
These operators satisfy the standard canonical commutation relations. The second quantization of a self-adjoint operator $A$ on $\mathfrak{H}$ will be denoted by 
\[
\sq(A)=0\oplus\bigoplus_{k\geq1}\sum_{j=1}^kA_j=\iint A(x,x')a^\ast_xa_{x'}\diff x\diff x',
\] 
where $A(x,x')$ is the kernel of $A$. 
In particular, the number operator is defined as $\mathcal{N}=\sq(\id)$.
We will denote the second quantization of a two-body operator $B$ on $\mathfrak{H}^2$ by
\[
	\sq_2(B)=0\oplus0\oplus\bigoplus_{k\geq2}\sum_{1\leq i<j\leq k}B_{ij}=\frac{1}{2}\iiiint B(x,y;x',y')a_x^\ast a_y^\ast a_{x'}a_{y'}\diff x\diff y\diff x'\diff y',
\]
where $B(x,y;x',y')$ denotes the kernel of $B$. 
Using this notation, we can extend the Hamiltonian $H_{N,t}$ to $\mathcal{F}$ by writing 
\begin{equation}\label{eq:Hsq}
	H_{N,t}=\sq(-\Delta+V_t)+\frac{1}{N-1}\sq_2(w_N),
\end{equation}
where we interpret $w_N$ as the two-body multiplication operator $w_N(x-y)$, which coincides with the definition \eqref{eq:Hamiltonian} on $\mathfrak{H}^N$. 
We can decompose the solution to the Schr{\"o}dinger equation \eqref{eq:linearschr} as
\[
	\Psi_N(t)=\sum_{k=0}^Nu_N(t)^{\otimes(N-k)}\otimes_\textnormal{s}\phi_N^{(k)}(t)=\sum_{k=0}^N\frac{(a^\ast(u_N(t)))^{N-k}}{\sqrt{(N-k)!}}\phi_N^{(k)}(t).
\]
Our goal is to approximate the vector
\begin{equation}\label{eq:phiN}
	\Phi_N(t)=\bigl(\phi_N^{(k)}(t)\bigr)_{k=0}^N\in\mathcal{F}^{\leq N}_\perp(t)
\end{equation}
by the solution of the Bogoliubov equation
\begin{equation}\label{eq:Bogoliubov}
	\begin{cases}
	\ii\partial_t\Phi(t)=\BH(t)\Phi(t), \\
	\Phi(0)=\Phi_0,
	\end{cases}
\end{equation}
where $\BH(t)$ is the Bogoliubov Hamiltonian
\[
	\BH(t)=\sq(h(t)+K_1(t))+\frac{1}{2}\biggl(\iint K_2(t,x,y)a^\ast_xa^\ast_y\diff x\diff y+\hc\biggr). 
\] 
Here, $h(t)$ is defined in \eqref{eq:Hartree:h} and we denote
\[
	K_1(t)=q(t)\tilde{K}_1(t)q(t),
	\qquad 
	K_2(t)=q(t)\otimes q(t)\tilde{K}_2(t),
\]
where
\[
	q(t)=\id-p(t)=\id-\ketbra{u_N(t)}{u_N(t)}
\]
and the kernel of the operator $\tilde{K}_1(t)$ and the function $\tilde{K}_2(t)\in\mathfrak{H}^2$ are given by 
\begin{align*}
	\tilde{K}_1(t,x,y)&=u_N(t,x)w_N(x-y)\overline{u_N(t,y)}, \\
	\tilde{K}_2(t,x,y)&=u_N(t,x)w_N(x-y)u_N(t,y). 
\end{align*}

We are now able to state our main theorem. 

\begin{thm}[Norm approximation]\label{thm:norm}
Let $w_N$ satisfy Assumption~\ref{assu:wN} and $V$ satisfy Assumption~\ref{assu:V}. 
Moreover, let $\Psi_N(t)\in L^2_\textnormal{s}(\R^{2N})$ be the solution to the Schr{\"o}dinger equation \eqref{eq:linearschr} with initial datum $\Psi_{N,0}\in L^2_\textnormal{s}(\R^{2N})$,
let $u_N$ be the solution to the Hartree equation \eqref{eq:Hartree} with initial datum $\varphi_0\in\Sigma^4$, $\norm{\varphi_0}=1$, and
let $\Phi(t)=(\phi^{(k)}(t))_{k=0}^\infty\in\mathcal{F}_\perp(t)$ be the solution to the Bogoliubov equation \eqref{eq:Bogoliubov} with initial datum $\Phi_0=(\phi^{(k)}_0)_{k=0}^\infty\in\mathcal{F}_\perp(0)$, $\norm{\Phi_0}=1$. Additionally, the initial data $\Psi_{N,0}$, $\varphi_0$, and $\Phi_0$ satisfy
\begin{equation}\label{eq:thm:norm:initial}
	\langle\Phi_0,\sq(1-\Delta+v)\Phi_0\rangle\leq C
\end{equation}
with $v$ from Assumption~\ref{assu:V} \ref{assu:Vd},
\[
	\Psi_{N,0}=\sum_{k=0}^N\varphi_0^{\otimes(N-k)}\otimes_\textnormal{s}\phi^{(k)}_0=\sum_{k=0}^N\frac{(a^\ast(\varphi_0))^{N-k}}{\sqrt{(N-k)!}}\phi^{(k)}_0,
\]
and 
\begin{equation}\label{eq:thm:norm:initial-ext}
	\langle \Psi_{N,0},\sum_{j=1}^N\abs{x_j}^s\Psi_{N,0}\rangle\leq CN
\end{equation}
for some $0<s\leq1$. 
\begin{enumerate}[(i)]
\item
If $w_N$ satisfies \eqref{eq:poly} with $\beta>0$, let $0<\alpha<1/3$. 
\item
If $w_N$ satisfies \eqref{eq:exp} with $0<\kappa<1/4$, let $0<\alpha<\min(1/2-2\kappa,1/3-7\kappa/6)$.
\end{enumerate}
Then for all $t\geq0$, there exists a constant $C_t>0$ independent of $N$ such that for large enough $N$, we have
\begin{equation}\label{eq:normapprox}
	\norm[\bigg]{\Psi_N(t) - \sum_{k=0}^Nu_N(t)^{\otimes(N-k)}\otimes_\textnormal{s}\phi^{(k)}(t)}^2 \leq C_t N^{-\alpha}. 
\end{equation}
\end{thm}

\begin{rems}
\begin{enumerate}
\item
In the case without an external potential ($V\equiv 0$), the assumption $\varphi_0\in\Sigma^4$ can be relaxed to $\varphi_0\in H^4(\R^2)$ and $\abs{x}^{s/2}\varphi_0\in L^2(\R^2)$, where $s$ is as in \eqref{eq:thm:norm:initial-ext}, by using the bounds from \cite[Lemma~10]{BosDieNam-26} for the solution of the Hartree equation. 
\item
For instance, the assumptions on the initial datum $\Psi_{N,0}$ are satisfied if
\[
	\Psi_{N,0} = \varphi_0^{\otimes N},
\]
where $\varphi_0\in\Sigma^4$, $\norm{\varphi_0}=1$. 
\item
The initial $N$-body state $\Psi_{N,0}$ is not normalized, but $\norm{\Psi_{N,0}}$ converges to $1$ very fast. In fact,
\begin{multline*}
	\norm{\Psi_{N,0}}^2 = \norm{\ind(\numop\leq N)\Phi_0}^2 = 1-\langle\Phi_0,\ind(\numop> N)\Phi_0\rangle \\
	\geq 1-N^{-1}\langle\Phi_0,\numop\Phi_0\rangle \geq 1-CN^{-1}
\end{multline*}
by \eqref{eq:thm:norm:initial}. We omit a normalization in the theorem, which would not affect the rate of convergence in \eqref{eq:normapprox}. 
\item 
It is sufficient that \eqref{eq:thm:norm:initial-ext} is satisfied for some $s>0$ since in this case $\norm{\Psi_{N,0}}\leq1$ implies it is also satisfied for some different $0<\tilde{s}\leq\min(1,s)$. The reason for requiring $0<s\leq1$ is a technical step in the proof of Lemma~\ref{lem:apkinetic}. 
\item
This theorem extends results achieved in \cite[Theorem~1]{NamNap-19} for $\beta\in(0,1)$ and in \cite[Theorem~5]{BosDieNam-26} for $\beta\in(0,3/2)$, but compared with \cite{BosDieNam-26} it requires stability in the sense that 
\[
	\int_{\R^2}w_- < a^\ast. 
\] 
The rate of convergence in \eqref{eq:normapprox} is improved compared to \cite[Theorem~1]{NamNap-19} and \cite[Theorem~5]{BosDieNam-26}, where $0<\alpha<(1-\beta)/3$ and $0<\alpha<\min(1/4,(3-2\beta)/8)$, respectively. 
\item
Let us assume $\widehat{w}$ is Lipschitz continuous. Then, it is possible to show 
\begin{equation*}
	\norm[\bigg]{\Psi_N(t) - \sum_{k=0}^N\varphi(t)^{\otimes(N-k)}\otimes_\textnormal{s}\phi_\textnormal{NLS}^{(k)}(t)}^2 \leq C_t N^{-\alpha'}
\end{equation*}
for $\alpha'=\min(\alpha,\beta-1)$ if $w_N$ satisfies \eqref{eq:poly} and $\alpha'=\alpha$ if $w_N$ satisfies \eqref{eq:exp}. Here, $\Phi_\textnormal{NLS}(t)=(\phi_\textnormal{NLS}^{(k)}(t))_{k=0}^\infty\in\mathcal{F}(\{\varphi(t)\}^\perp)$ is the solution to the modified Bogoliubov equation
\begin{equation*}
	\begin{cases}
	\ii\partial_t\Phi_\textnormal{NLS}(t)=\BH_\textnormal{NLS}(t)\Phi_\textnormal{NLS}(t), \\
	\Phi_\textnormal{NLS}(0)=\Phi_0,
	\end{cases}
\end{equation*}
with 
\begin{align*}
	\BH_\textnormal{NLS}(t)&=\sq(h_\textnormal{NLS}(t)+K_{\textnormal{NLS},1}(t))+\frac{1}{2}\biggl(\iint K_{\textnormal{NLS},2}(t,x,y)a^\ast_xa^\ast_y\diff x\diff y+\hc\biggr), \\
	h_\textnormal{NLS}(t)&=-\Delta_x+V(t,x)+a\abs{\varphi(t,x)}^2-\mu(t),
\end{align*}
and $K_{\textnormal{NLS},1}$ and $K_{\textnormal{NLS},2}$ are defined like $K_1$ and $K_2$ with $\varphi$ instead of $u_N$. This can be achieved by defining the unitary map $U_N(t)$ in Section~\ref{sec:app} via $\varphi$ instead of $u_N$. Proceeding with this map produces additional error terms such as 
\[
	(N-\numop)\langle\varphi(t),(w_N\ast\abs{\varphi(t)}^2-a\abs{\varphi(t)}^2)\varphi(t)\rangle.
\]
With the bound
\begin{align*}
	\norm{(w_N\ast\abs{\varphi(t)}^2-a\abs{\varphi(t)}^2)\varphi(t)}_{L^2} &\leq \norm{\varphi(t)}_{L^\infty} \norm[\big]{\widehat{w_N}\widehat{\abs{\varphi(t)}^2}-\widehat{w}(0)\widehat{\abs{\varphi(t)}^2}}_{L^2} \\
	&\leq \norm{\varphi(t)}_{L^\infty} \norm[\bigg]{\frac{\widehat{w}(N^{-\beta}\cdot)-\widehat{w}(0)}{\abs{\cdot}}}_{L^\infty} \norm[\big]{\abs{\cdot} \widehat{\abs{\varphi(t)}^2}}_{L^2} \\
	&\leq C\norm{\varphi(t)}_{L^\infty}^2\norm{\nabla\varphi(t)}_{L^2}N^{-\beta},
\end{align*}
where the last step follows from the assumption that $\widehat{w}$ is Lipschitz, it can be proved that these terms lead to errors of order $\mathcal{O}(N^{1-\beta})$ if $w_N$ satisfies \eqref{eq:poly} and similarly, $\mathcal{O}(N\ee^{-N^\kappa})$ if $w_N$ satisfies \eqref{eq:exp}. To simplify the presentation, we do not include this extension. 
\end{enumerate}
\end{rems}

As explained in \cite{LewNamSch-15}, the norm convergence in Theorem~\ref{thm:norm} implies the convergence in trace class of the reduced density matrix of $\Psi_N$, namely, we have the following result:

\begin{cor}[NLS evolution of the condensate]\label{cor:condensate}
Under the same assumptions as in Theorem~\ref{thm:norm}, we have the convergence in trace class
\[
	\lim_{N\to\infty} \gamma^{(1)}_{\Psi_N(t)} = \ketbra{\varphi(t)}{\varphi(t)}, \qquad \forall t\geq0,
\]
where $\varphi$ is the solution to the nonlinear Schr{\"o}dinger equation \eqref{eq:NLS} with initial datum $\varphi_0$. 
\end{cor}

This extends results obtained in \cite{CheHol-17,JebPic-18,NamNap-19,BosDieNam-26}. 

\subsection{Roadmap of the proof}
To prove Theorem~\ref{thm:norm}, we estimate $\norm{\Phi_N(t)-\Phi(t)}$, where $\Phi_N(t)$ is defined in \eqref{eq:phiN} and $\Phi(t)$ is the Bogoliubov dynamics defined in \eqref{eq:Bogoliubov}. We start by establishing uniform-in-$N$ estimates for the Hartree dynamics in Section~\ref{sec:NLS}. 

In Section~\ref{sec:Bog}, we prove an upper bound on the kinetic energy of the Bogoliubov dynamics, which improves the corresponding estimate in \cite[Lemma~5]{NamNap-19}. Namely, instead of a bound of the form
\[
	\langle\Phi(t),\sq(1-\Delta)\Phi(t)\rangle \leq C_{t,\varepsilon}N^{\varepsilon\beta}
\]
for all $\varepsilon>0$, which is proved in \cite{NamNap-19}, we provide an improved estimate of the form
\[
	\langle\Phi(t),\sq(1-\Delta)\Phi(t)\rangle \leq C_{t}\beta\log(N). 
\]
This improvement is necessary to handle the case where $w_N$ satisfies \eqref{eq:exp}, which corresponds to $\beta=N^\kappa/\log(N)$. We also prove a bound on the particle number of the Bogoliubov dynamics
\[
	\langle\Phi(t),\numop\Phi(t)\rangle \leq C_{t}.
\]
Since we do not require the initial data $\Phi_0$ to be quasi-free, this is an improvement over the corresponding result in \cite[Proposition~7]{NamNap-17a}. 

The second ingredient to be able to cover the case where $w_N$ satisfies \eqref{eq:exp} is the improved operator estimate on the interaction $w_N(x-y)$
\begin{equation}\label{eq:MT:rmotp}
	\abs{w_N(x-y)}\leq C\beta\log(N)(1-\Delta_x),
\end{equation}
which we will prove in Lemma~\ref{lem:MT} in Section~\ref{sec:app}. This bound will be used to obtain better error estimates for the Bogoliubov approximation. To obtain this bound, we use the Moser--Trudinger inequality. This inequality is an improvement over the standard Sobolev inequality, which only allows one to bound
\[
	\abs{w_N(x-y)}\leq C_s\norm{w_N}_{L^s}(1-\Delta_x)=C_sN^{2\beta(1-1/s)}\norm{w}_{L^s}(1-\Delta_x), \qquad s>1.
\]
In addition, we will prove the new bound for the kinetic energy of the full dynamics
\[
	\langle\Phi_N(t),\sq(1-\Delta)\Phi_N(t)\rangle\leq C_t N
\]
in Lemma~\ref{lem:apkinetic} in Section~\ref{sec:app}. We assume the initial datum $\Psi_{N,0}$ of the dynamics to be localized, 
\begin{equation*}
	\langle \Psi_{N,0},\sum_{j=1}^N\abs{x_j}^s\Psi_{N,0}\rangle\leq CN
\end{equation*}
for $0<s\leq1$. Hence, we can use the fact that in some sense $\Psi_N(t)$ cannot spread too quickly, i.e., we can control the quantity
\begin{equation*}
	\langle \Psi_{N}(t),\sum_{j=1}^N\abs{x_j}^s\Psi_{N}(t)\rangle. 
\end{equation*}
This allows us to use the stability of second kind of the modified Hamiltonian 
\[
	H_{N,t} + \sum_{j=1}^N\abs{x_j}^s,
\]
which follows from \cite[Theorem~2]{JunVis-26}, to bound the kinetic energy of $\Psi_N(t)$ and therefore also the kinetic energy of $\Phi_N(t)$. 
The new bound on the kinetic energy of $\Phi_N(t)$ allows us to cover any $0<\beta<\infty$ since the restriction $\beta<1$ in \cite{NamNap-19} came from a weaker bound on the kinetic energy of $\Phi_N(t)$ of order $N+N^{2\beta}$.
The bounds on the kinetic energy of $\Phi_N(t)$ and $\Phi(t)$ together with the estimate \eqref{eq:MT:rmotp} are used to bound terms containing the interaction $w_N(x-y)$.

In Section~\ref{sec:ftb}, instead of comparing $\Phi_N(t)$ and $\Phi(t)$ directly, we first compare $\Phi_N(t)$ with a truncated dynamics $\Phi_{N,M_1}(t)$, whose generator is partially truncated to the sector of at most $M_1$ particles, $M_1\ll N$, which we will compare with a second truncated dynamics $\Phi_{N,M_2}(t)$, $M_2\ll M_1$. Then, we will compare $\Phi_{N,M_2}(t)$ with the Bogoliubov dynamics $\Phi(t)$.
We choose a slightly different truncation than in \cite[Section~6]{NamNap-19}. Instead of truncating the entire generator $\mathcal{G}_N(t)$ of $\Phi_N(t)$ to the sector of at most $M$ particles, we only truncate the error term $\mathcal{G}_N(t)-\BH(t)$ using a smooth cut-off function $\chi_M$. This enables us to use the proof of the bound on the particle number of the Bogoliubov dynamics up to bounding error terms, which come from $\mathcal{G}_N(t)-\BH(t)$, to prove a bound of the form
\[
	\langle\Phi_{N,M}(t),\numop\Phi_{N,M}(t)\rangle \leq C_{t}
\]
if $M$ is small enough compared to $N$. The smoothness of $\chi_M$ is used to bound some of the error terms in this proof. 
Introducing two truncated dynamics instead of one as in \cite{NamNap-19} or \cite{BosDieNam-26} yields a better rate of convergence and a slightly larger range of $\kappa$ in Theorem~\ref{thm:norm}.  

Finally, we conclude the proof of our main results, Theorem~\ref{thm:norm} and Corollary~\ref{cor:condensate}, in Section~\ref{sec:con}. 

\section{Nonlinear Schr{\"o}dinger and Hartree equations}\label{sec:NLS}

The following result on the global existence of the solutions to the NLS \eqref{eq:NLS} and Hartree equations \eqref{eq:Hartree} is a simple consequence of \cite[Proposition~2.5 and Corollary~2.6]{Car-11} and the proof of \cite[Theorem~1.3]{CarSil-15}. 

\begin{lem}[Global existence and uniqueness of $\varphi$ and $u_N$]\label{lem:NLSregularity}
Let $w_N$ satisfy Assumption~\ref{assu:wN} and let $V$ satisfy Assumption~\ref{assu:V} \ref{assu:Va}, \ref{assu:Vb}, \ref{assu:Vc}. For $\varphi_0\in\Sigma^1$ with $\norm{\varphi_0}_{L^2}=1$, there exists a unique solution $\varphi$ to \eqref{eq:NLS} such that
\[
	\varphi,\nabla\varphi,x\varphi\in C\bigl([0,\infty);L^2(\R^2)\bigr)\cap L^4_\textnormal{loc}\bigl([0,\infty);L^4(\R^2)\bigr)
\]
and a unique solution $u_N$ to \eqref{eq:Hartree} such that
\[
	u_N,\nabla u_N,xu_N\in C\bigl([0,\infty);L^2(\R^2)\bigr)\cap L^4_\textnormal{loc}\bigl([0,\infty);L^4(\R^2)\bigr).
\]
Moreover, there exists $C>0$ such that 
\[
	\norm{\varphi(t)}_{\Sigma^1}\leq C\exp(Ct), \qquad \norm{u_N(t)}_{\Sigma^1}\leq C\exp(Ct). 
\]
\end{lem}

\begin{proof}
For convenience, we will work with the equation
\begin{equation}\label{eq:NLS2}
\begin{cases}
	\ii\partial_t\varphi(t,x)=\bigl(-\Delta_x+V(t,x)+a\abs{\varphi(t,x)}^2\bigr)\varphi(t,x), \\
	\varphi(0,x)=\varphi_0(x).
\end{cases}
\end{equation}
The solutions of \eqref{eq:NLS} and \eqref{eq:NLS2} are related via the gauge transformation
\[
	\varphi\mapsto \ee^{-\ii\int_0^t\mu(t')\diff t'}\varphi,
\]
where $\mu(t)$ is defined as in \eqref{eq:NLS}. 
The local existence of $\varphi$ follows from \cite[Proposition~2.5]{Car-11}. 
We have mass conservation
\[
	\frac{\diff}{\diff t}\norm{\varphi(t)}_{L^2}^2=0
\]
and thus, $\norm{\varphi(t)}_{L^2}=\norm{\varphi(0)}_{L^2}=1$. 
As in the proof of \cite[Theorem~1.3]{CarSil-15}, we can define the pseudo-energy
\[
	\mathcal{E}(t)=\norm{\nabla\varphi(t)}^2_{L^2}+\frac{a}{2}\norm{\varphi(t)}^4_{L^4}+\frac{1}{2}\int_{\R^2}\abs{x}^2\abs{\varphi(t,x)}^2\diff x. 
\]
Using 
\[
	a\geq -\int_{\R^2}w_->-a^\ast,
\]
the Gagliardo--Nirenberg inequality \eqref{eq:GN}, and $\norm{\varphi(t)}_{L^2}=1$, we obtain the following lower bound on $\mathcal{E}$
\begin{equation}\label{eq:NLS:Elower}
\begin{aligned}
	\mathcal{E}(t)&\geq\norm{\nabla\varphi(t)}^2_{L^2}\biggl(1-\frac{1}{a^\ast}\int_{\R^2}w_-\biggr)+\frac{1}{2}\int_{\R^2}\abs{x}^2\abs{\varphi(t,x)}^2\diff x \\
	&\geq C(\norm{\nabla\varphi(t)}^2_{L^2}+\norm{x\varphi(t)}_{L^2}^2)\geq C\norm{\nabla\varphi(t)}_{L^2}\norm{x\varphi(t)}_{L^2}.
\end{aligned}
\end{equation}
As in \cite{CarSil-15}, we compute 
\begin{align*}
	\frac{\diff\mathcal{E}}{\diff t}(t) 
	&=\Re\Bigl\langle\bigl(-2\Delta\varphi(t)+2a\abs{\varphi(t)}^2\varphi(t)+\abs{x}^2\varphi(t)\bigr),\partial_t\varphi(t)\Bigr\rangle \\
	&=\Re\int_{\R^2}\overline{\varphi(t,x)}(\abs{x}^2-2V(t,x))\partial_t\varphi(t,x)\diff x \\
	&=\Im\int_{\R^2}\overline{\varphi(t,x)}(\abs{x}^2-2V(t,x))(-\Delta\varphi(t,x))\diff x \\
	&=2\Im\int_{\R^2}\overline{\varphi(t,x)}(x-\nabla V(t,x))\cdot\nabla\varphi(t,x)\diff x,
\end{align*}
where we used \eqref{eq:NLS2} in the second and third equalities and integrated by parts in the last equality. 
Thus, using Assumption~\ref{assu:V} \ref{assu:Vc}, the Cauchy--Schwarz inequality, and the bound \eqref{eq:NLS:Elower}, we obtain
\[
	\frac{\diff\mathcal{E}}{\diff t}(t)\leq C(1+\norm{\nabla\varphi(t)}_{L^2}+\norm{x\varphi(t)}_{L^2}\norm{\nabla\varphi(t)}_{L^2})\leq C(1+\mathcal{E}(t)).
\]
Gr{\"o}nwall's lemma together with \eqref{eq:NLS:Elower} yields the bound 
\[
\norm{\nabla\varphi(t)}_{L^2}+\norm{x\varphi(t)}_{L^2}\leq C\exp(Ct).
\]
Hence, the proof of \cite[Corollary~2.6]{Car-11} shows global existence of $\varphi$, and since $\norm{\varphi(t)}_{L^2}=1$, we also have $\norm{\varphi(t)}_{\Sigma^1}\leq C\exp(Ct)$. 

From a straightforward adaptation of these arguments follow the statements about the Hartree solution $u_N$: Instead of $\mathcal{E}(t)$, we define the Hartree pseudo-energy
\[
	\mathcal{E}_N(t)=\norm{\nabla u_N(t)}^2_{L^2}+\frac{1}{2}\int_{\R^2}(w_N\ast\abs{u_N(t,\cdot)}^2)(x)\abs{u_N(t,x)}^2\diff x+\frac{1}{2}\int_{\R^2}\abs{x}^2\abs{u_N(t,x)}^2\diff x. 
\]
Using Young's inequality, we estimate
\begin{align*}
\int_{\R^2}(w_N\ast\abs{u_N(t,\cdot)}^2)(x)\abs{u_N(t,x)}^2\diff x
&\geq -\int_{\R^2}([w_N]_-\ast\abs{u_N(t,\cdot)}^2)(x)\abs{u_N(t,x)}^2\diff x \\
&\geq -\norm{[w_N]_-\ast\abs{u_N(t)}^2}_{L^2}\norm{u_N(t)}_{L^4}^2 \\
&\geq -\norm{u_N(t)}_{L^4}^4\int_{\R^2}w_-.
\end{align*}
The rest of the proof stays the same and we obtain global existence of $u_N$ and the $N$-independent bound $\norm{u_N(t)}_{\Sigma^1}\leq C\exp(Ct)$. 
\end{proof}

The proof of \cite[Corollary~1.4]{CarSil-15} can be adapted to prove the following bounds on the solution of the Hartree equation $u_N$. Note that none of the constants in the lemma depend on $N$. 

\begin{lem}\label{lem:Hregularity}
Let $w_N$ satisfy Assumption~\ref{assu:wN} and let $V$ satisfy Assumption~\ref{assu:V} \ref{assu:Va}, \ref{assu:Vb}, \ref{assu:Vc}. For $\varphi_0\in\Sigma^4$ with $\norm{\varphi_0}_{L^2}=1$, there exists a unique solution $u_N$ to \eqref{eq:Hartree} such that for all $t\geq0$
\begin{align*}
	\norm{u_N(t)}_{L^\infty}&\leq C\norm{u_N(t)}_{H^2}\leq C\norm{u_N(t)}_{\Sigma^4}\leq C\exp(C\exp(Ct)), \\
	\norm{\partial_tu_N(t)}_{L^\infty}&\leq C\norm{\partial_tu_N(t)}_{H^2}\leq C\exp(C\exp(Ct)). 
\end{align*}
\end{lem}

\begin{proof}
As in the proof of Lemma~\ref{lem:NLSregularity}, we will work with the equation
\begin{equation}\label{eq:Hartree2}
\begin{cases}
	\ii\partial_tu_N(t,x)=\bigl(-\Delta_x+V(t,x)+(w_N\ast\abs{u_N(t,\cdot)}^2)(x)\bigr)u_N(t,x), \\
	u_N(0,x)=\varphi_0(x)
\end{cases}
\end{equation}
instead of \eqref{eq:Hartree}. 

The proof consists of two steps. We first bound $\norm{u_N(t)}_{\Sigma^4}$ and then we prove $\norm{\partial_tu_N(t)}_{H^2}\leq C\norm{u_N(t)}_{\Sigma^4}^3$. The remaining estimates follow from the trivial inequality $\norm{\cdot}_{H^2}\leq\norm{\cdot}_{\Sigma^4}$ and the Sobolev inequality $\norm{\cdot}_{L^\infty}\leq C\norm{\cdot}_{H^2}$. 

\paragraph{Step 1: $\norm{u_N(t)}_{\Sigma^4}\leq C\exp(C\exp(Ct))$.}

This bound follows from the method in the proof of \cite[Corollary~1.4]{CarSil-15}: The bound $\norm{u_N(t)}_{\Sigma^1}\leq C\exp(Ct)$ from Lemma~\ref{lem:NLSregularity} is the starting point of an induction procedure, where Strichartz estimates are repeatedly used. 
It is easy to see that this can be done in a way that the error terms depend on $w_N$ only via its $L^1$-norm $\norm{w_N}_{L^1}=\norm{w}_{L^1}$, which is independent of $N$, by using Young's inequality in the form
\[
	\norm{w_N\ast f}_{L^r} \leq \norm{w_N}_{L^1} \norm{f}_{L^r}, \qquad r\geq1. 
\]
Thus, all bounds are independent of $N$. 

\paragraph{Step 2: $\norm{\partial_tu_N(t)}_{H^2}\leq C\norm{u_N(t)}_{\Sigma^4}^3$.}

Using \eqref{eq:Hartree2}, we obtain
\begin{equation*}
	\norm{\partial_tu_N(t)}_{H^2}\leq \norm{-\Delta u_N(t)}_{H^2}+\norm{V_tu_N(t)}_{H^2}+\norm{(w_N\ast\abs{u_N(t,\cdot)}^2)u_N(t)}_{H^2}. 
\end{equation*}
By Assumption~\ref{assu:V} \ref{assu:Vc} we have the bound $\abs{V(t,x)}\leq C\jbr{x}^2$. Thus, we can estimate $\norm{-\Delta u_N(t)}_{H^2}\leq\norm{u_N(t)}_{\Sigma^4}$, $\norm{V_tu_N(t)}_{H^2}\leq C\norm{u_N(t)}_{\Sigma^4}$, and 
\begin{align*}
	\norm{(w_N\ast\abs{u_N(t,\cdot)}^2)u_N(t)}_{H^2}
	&\leq C\norm{w}_{L^1}\norm{u_N(t)}_{L^4}^2\norm{u_N(t)}_{L^\infty} \\
	&\quad+C\norm{w}_{L^1}\norm{\nabla u_N(t)}_{L^4}^2\norm{u_N(t)}_{L^\infty} \\
	&\quad+C\norm{w}_{L^1}\norm{-\Delta u_N(t)}_{L^2}\norm{u_N(t)}_{L^\infty}^2 \\
	&\leq C\norm{u_N(t)}_{H^2}^3,
\end{align*}
where we used the Sobolev inequalities $\norm{\cdot}_{L^\infty}\leq C\norm{\cdot}_{H^2}$ and $\norm{\cdot}_{L^4}\leq C\norm{\cdot}_{H^1}$. 
Therefore, we have $\norm{\partial_tu_N(t)}_{H^2}\leq C\norm{u_N(t)}_{\Sigma^4}^3$. 
\end{proof}

\section{Bogoliubov equation}\label{sec:Bog}

The following lemma, which is an improvement of \cite[Lemma~5]{NamNap-19}, is the main result of this section. The particle number bound \eqref{eq:number:bog} has previously been proved for quasi-free initial data in \cite[Proposition~7]{NamNap-17a}. We provide a proof that does not require the initial data to be quasi-free. 
\begin{lem}[Bogoliubov dynamics]\label{lem:bogdyn} 
Let $w_N$ satisfy Assumption~\ref{assu:wN} and let $V$ satisfy Assumption~\ref{assu:V}. Let $u_N$ be the Hartree evolution in \eqref{eq:Hartree} with initial datum $\varphi_0\in\Sigma^4$, $\norm{\varphi_0}_{L^2}=1$, and let $\Phi_0\in\mathcal{F}_\perp(0)$ be a normalized vector satisfying \eqref{eq:thm:norm:initial}. Then, the Bogoliubov equation \eqref{eq:Bogoliubov} with initial datum $\Phi_0$ has a unique global solution 
\[
	\Phi\in C([0,\infty),\mathcal{F})\cap L^\infty_\textnormal{loc}((0,\infty),\mathcal{Q}(\sq(1-\Delta+v)))
\]
and the following statements hold. 
\begin{enumerate}[(i)]
\item\label{item:kinetic:bog}
If the interaction potential $w_N$ satisfies \eqref{eq:poly} for $\beta>0$, then we have
\begin{equation}\label{eq:kinetic:bog:poly}
	\langle\Phi(t),\sq(1-\Delta+v)\Phi(t)\rangle\leq C_{t}\beta\log(N). 
\end{equation}
If $w_N$ satisfies \eqref{eq:exp} for $0<\kappa<1$, then we have
\begin{equation}\label{eq:kinetic:bog:exp}
	\langle\Phi(t),\sq(1-\Delta+v)\Phi(t)\rangle\leq C_{t}N^\kappa.
\end{equation}
\item\label{item:number:bog}
The particle number is bounded,
\begin{equation}\label{eq:number:bog}
	\langle\Phi(t),\numop\Phi(t)\rangle \leq C_t. 
\end{equation}
\end{enumerate}
\end{lem}

We proceed as in \cite{NamNap-19} and use the following lemma, which is an improvement of \cite[Lemma~6]{NamNap-19}, to deduce Lemma~\ref{lem:bogdyn} \ref{item:kinetic:bog}. 
\begin{lem}\label{lem:bogbounds}
For every $\eta>0$ we have
\begin{align*}
	\pm\bigl(\BH(t)+\sq(\Delta-V_t)\bigr)&\leq\eta\sq(1-\Delta)+C_{t}(1+\eta^{-1})\bigl(\numop+\beta\log(N)\bigr), \\
	\pm\partial_t\BH(t)&\leq C_t\sq(1-\Delta+v)+C_{t}\bigl(\numop+\beta\log(N)\bigr), \\
	\pm\ii[\BH(t),\numop]&\leq\eta\sq(1-\Delta)+C_{t}(1+\eta^{-1})\bigl(\numop+\beta\log(N)\bigr)
\end{align*}
if $w_N$ satisfies \eqref{eq:poly} for $\beta>0$. If $w_N$ satisfies \eqref{eq:exp} for $0<\kappa<1$, $\beta\log(N)$ is replaced by $N^{\kappa}$. 
\end{lem}

Using \eqref{eq:sqrtdeltaK2bound} from the following lemma, which improves \cite[Lemma~8]{NamNap-19}, Lemma~\ref{lem:bogbounds} can be proved similarly to \cite[Lemma~6]{NamNap-19} and we skip the details of the proof. The estimates \eqref{eq:deltaxK2bound} and \eqref{eq:kperp:infty2bound} will be used in the proofs of Lemma~\ref{lem:bogdyn} \ref{item:number:bog} and Lemma~\ref{lem:truncatednumber}, respectively. 

\begin{lem}\label{lem:K2bounds}
Let $L\geq1$ and denote 
\[
	k_\perp(t) = q(t)\otimes q(t) k(t), \qquad k(t)=(2L-\Delta_x-\Delta_y)^{-1}K_2(t). 
\]
Then, we have for $r\geq1$ and $\bm{\alpha}\in\N_0^2$, $\abs{\bm{\alpha}}\leq2$,
\begin{align}
	\norm{(1-\Delta_x)^{-1/2}K_2(t,\cdot,\cdot)}_{L^2}^2+\norm{(1-\Delta_x)^{-1/2}\partial_tK_2(t,\cdot,\cdot)}_{L^2}^2 &\leq C_{t}\beta\log(N), \label{eq:sqrtdeltaK2bound} \\
	\norm{(2L-\Delta_x-\Delta_y)^{-r}x^{\bm{\alpha}}K_2(t,\cdot,\cdot)}_{L^2}^2+\norm{(2L-\Delta_x-\Delta_y)^{-r}\partial_tK_2(t,\cdot,\cdot)}_{L^2}^2 &\leq C_tL^{1-2r}, \label{eq:deltaxK2bound} \\
	\norm{k_\perp(t)}_{L^\infty_{1}L^2_{2}} \coloneqq \sup_{x\in\R^2}\norm{k_\perp(t,x,\cdot)}_{L^2} &\leq C_t \label{eq:kperp:infty2bound}
\end{align}
if $w_N$ satisfies \eqref{eq:poly} for $\beta>0$. If $w_N$ satisfies \eqref{eq:exp} for $0<\kappa<1$, $\beta\log(N)$ is replaced by $N^{\kappa}$. 
\end{lem}

\begin{proof}
We focus on the case where $w_N$ satisfies \eqref{eq:poly}, since the case where $w_N$ satisfies \eqref{eq:exp} follows by setting $\beta=N^\kappa/\log(N)$. 

\paragraph{Proof of \eqref{eq:sqrtdeltaK2bound} and \eqref{eq:deltaxK2bound}.}

Let $L,r\geq1$ and $\bm{\alpha}\in\N_0^2$, $\abs{\bm{\alpha}}\leq2$. 
We restrict ourselves to the bounds for $K_2$ since the bounds for $\partial_t K_2$ can be proved similarly due to the bounds on $\partial_tu_N(t)$ in Lemma~\ref{lem:Hregularity}. 
Recall the definition of $K_2$,
\[
	K_2(t)=q(t)\otimes q(t)\tilde{K}_2(t), \qquad \tilde{K}_2(t,x,y)=u_N(t,x)w_N(x-y)u_N(t,y),
\]
where $q(t)=\id-p(t)=\id-\ketbra{u_N(t)}{u_N(t)}$. We have $\id-q(t)\otimes q(t)=p(t)\otimes\id+\id\otimes p(t)-p(t)\otimes p(t)$. Using the bounds from Lemma~\ref{lem:Hregularity}, we find
\begin{align*}
	\MoveEqLeft[3] \norm{x^{\bm{\alpha}}(\ketbra{u_N(t)}{u_N(t)}\otimes\id)\tilde{K}_2(t,\cdot,\cdot)}^2_{L^2} \\
	&=\iint\abs*{\int \overline{u_N(t,z)}u_N(t,z)w_N(z-y)\diff z}^2\abs{x^{\bm{\alpha}}u_N(t,x)}^2\abs{u_N(t,y)}^2\diff x\diff y \\
	&\leq \norm{u_N(t,\cdot)}_{L^\infty}^4\norm{w_N}_{L^1}^2\norm{u_N(t,\cdot)}_{\Sigma^2}^2\norm{u_N(t,\cdot)}_{L^2}^2\leq C_t,
\end{align*}
and similarly, $\norm{x^{\bm{\alpha}}(\id\otimes p(t))\tilde{K}_2(t,\cdot,\cdot)}_{L^2}^2\leq C_t$ and $\norm{x^{\bm{\alpha}}(p(t)\otimes p(t))\tilde{K}_2(t,\cdot,\cdot)}_{L^2}^2\leq C_t$. Hence, we obtain
\begin{equation}\label{eq:K2tildeK2}
	\norm{x^{\bm{\alpha}}K_2(t,\cdot,\cdot)-x^{\bm{\alpha}}\tilde{K}_2(t,\cdot,\cdot)}_{L^2}=\norm{x^{\bm{\alpha}}(\id-q(t)\otimes q(t))\tilde{K}_2(t,\cdot,\cdot)}_{L^2}\leq C_t. 
\end{equation}
Since $(1-\Delta_x)^{-1}\leq\id$ and $(2L-\Delta_x-\Delta_y)^{-2r}\leq (2L)^{-2r}$, this bound implies
\[
	\norm{(1-\Delta_x)^{-1/2}K_2(t,\cdot,\cdot)-(1-\Delta_x)^{-1/2}\tilde{K}_2(t,\cdot,\cdot)}_{L^2}\leq C_t,
\]
where we used \eqref{eq:K2tildeK2} with $\bm{\alpha}=(0,0)$, and
\begin{multline*}
	\norm{(2L-\Delta_x-\Delta_y)^{-r}x^{\bm{\alpha}}K_2(t,\cdot,\cdot)-(2L-\Delta_x-\Delta_y)^{-r}x^{\bm{\alpha}}\tilde{K}_2(t,\cdot,\cdot)}_{L^2}^2 \\
	\leq (2L)^{-2r}\norm{x^{\bm{\alpha}}K_2(t,\cdot,\cdot)-x^{\bm{\alpha}}\tilde{K}_2(t,\cdot,\cdot)}_{L^2}^2 
	\leq C_t L^{-2r}. 
\end{multline*}
It remains to bound $\norm{(1-\Delta_x)^{-1/2}\tilde{K}_2(t,\cdot,\cdot)}_{L^2}$ and $\norm{(2L-\Delta_x-\Delta_y)^{-r}x^{\bm{\alpha}}\tilde{K}_2(t,\cdot,\cdot)}_{L^2}$. 
To simplify the notation, let us define $f_{\bm{\alpha}}(t,x,y)=x^{\bm{\alpha}}\tilde{K}_2(t,x,y)$ and $f=f_{\bm{\alpha}}$ if $\bm{\alpha}=(0,0)$. As in the proof of \cite[Lemma~8]{NamNap-19}, we have the following formula for the Fourier transform of $f_{\bm{\alpha}}$:
\[
	\widehat{f_{\bm{\alpha}}}(t,p,q)=\int w_N(z)\widehat{\bigl((x^{\bm{\alpha}}u_{N})_zu_N\bigr)}(t,p+q)\ee^{-2\pi\ii p\cdot z}\diff z,
\]
where $(x^{\bm{\alpha}}u_{N})_z(t,x)=(z+x)^{\bm{\alpha}}u_N(t,z+x)$. By Lemma~\ref{lem:Hregularity}, we obtain
\begin{align*}
	\int_{\R^2} \abs{\widehat{f_{\bm{\alpha}}}(t,p,q)}^2\diff q &\leq \norm{w_N}_{L^1}^2\sup_{z\in\R^2}\norm[\big]{\bigl((x^{\bm{\alpha}}u_{N})_zu_N\bigr)(t,\cdot)}^2_{L^2} \\
	&\leq \norm{w}_{L^1}^2\norm{u_N(t)}_{\Sigma^2}^2\norm{u_N(t)}_{L^\infty}^2 
	\leq C_t
\end{align*}
for every $p$. 
Using this bound, the Plancherel theorem, and the sets $\mathcal{L}_N=\{p:\abs{p}\leq N^\beta\}\subset\R^2$ and $\mathcal{H}_N=\R^2\setminus\mathcal{L}_N=\{p:\abs{p}>N^\beta\}$, we compute 
\begin{align*}
	\MoveEqLeft[3] \norm{(1-\Delta_x)^{-1/2}f(t,\cdot,\cdot)}_{L^2}^2=\iint(1+\abs{2\pi p}^2)^{-1}\abs{\widehat{f}(t,p,q)}^2\diff p\diff q \\
	&=\int_{\mathcal{L}_N}\biggl(\int_{\R^2} \abs{\widehat{f}(t,p,q)}^2\diff q\biggr)(1+\abs{2\pi p}^2)^{-1}\diff p \\
	&\quad+\int_{\R^2}\biggl(\int_{\mathcal{H}_N} (1+\abs{2\pi p}^2)^{-1}\abs{\widehat{f}(t,p,q)}^2\diff p\biggr)\diff q \\
	&\leq C_t\int_{\mathcal{L}_N}(1+\abs{2\pi p}^2)^{-1}\diff p + N^{-2\beta}\iint\abs{\widehat{f}(t,p,q)}^2\diff p\diff q \\
	&\leq C_t\int_0^{N^\beta}\frac{s}{1+s^2}\diff s + N^{-2\beta} \norm{f(t,\cdot,\cdot)}_{L^2}^2 \\
	&\leq C_t\beta\log(N) + N^{-2\beta}\norm{u_N(t)}_{L^2}^2\norm{w_N}_{L^2}^2\norm{u_N(t)}_{L^\infty}^2 \leq C_t\beta\log(N) + C_t. 
\end{align*}
This completes the proof of \eqref{eq:sqrtdeltaK2bound}. 
Similarly, we obtain 
 \begin{align*}
	\norm{(2L-\Delta_x-\Delta_y)^{-r}f_{\bm{\alpha}}(t,\cdot,\cdot)}_{L^2}^2
	&\leq\int\biggl(\int\abs{\widehat{f_{\bm{\alpha}}}(t,p,q)}^2\diff q\biggr)(2L+\abs{2\pi p}^2)^{-2r}\diff p \\
	&\leq C_t L^{1-2r}\int(2+\abs{2\pi p}^2)^{-2r}\diff p 
	\leq C_t L^{1-2r},
\end{align*}
which completes the proof of \eqref{eq:deltaxK2bound}. 

\paragraph{Proof of \eqref{eq:kperp:infty2bound}.}

To prove this bound, we decompose 
\[
	k_\perp(t) = k(t) - p(t)\otimes\id k(t) - q(t)\otimes p(t) k(t)
\]
and use
\begin{align*}
	\norm{p(t)\otimes\id k(t)}_{L^\infty_{1}L^2_{2}} &\leq \norm{u_N(t)}_{L^\infty}\norm{u_N(t)}_{L^2}\norm{k(t)}_{L^2}\leq C_t, \\
	\norm{q(t)\otimes p(t) k(t)}_{L^\infty_{1}L^2_{2}} &\leq \norm{k(t)}_{L^\infty_{1}L^2_{2}} + \norm{u_N(t)}_{L^\infty}\norm{u_N(t)}_{L^2}\norm{k(t)}_{L^2} \\
	&\leq \norm{k(t)}_{L^\infty_{1}L^2_{2}} + C_t. 
\end{align*}
Thus, it suffices to bound $\norm{k(t)}_{L^\infty_{1}L^2_{2}}$. We use the integral representation for the resolvent $(2L-\Delta_x-\Delta_y)^{-1}$ (see e.g.\ \cite[Chapter~II]{EngNag-00})
\begin{align*}
	k(t,x,y) &= (2L-\Delta_x-\Delta_y)^{-1}K_2(t,x,y) \\
	&= \int_0^\infty\int_{\R^2}\int_{\R^2} \ee^{-2Ls} F(s,x-x') F(s,y-y') K_2(t,x',y')\diff x'\diff y' \diff s,
\end{align*}
where we denoted the heat kernel by 
\[
	F(s,x)=(4\pi s)^{-1} \ee^{-\abs{x}^2/4s}. 
\]
We further decompose $K_2 = \tilde{K}_2 + (K_2-\tilde{K}_2)$ and, using Minkowski's integral inequality and Young's inequality, estimate
\begin{align*}
	\norm{k(t,x,\cdot)}_{L^2} &\leq \int_0^\infty\int_{\R^2} \ee^{-2Ls} F(s,x-x') \norm{F(s,\cdot)\ast K_2(t,x',\cdot)}_{L^2}\diff x'\diff s \\
	&\leq \int_0^\infty\ee^{-2Ls} \norm{F(s)}_{L^1} \Bigl( \norm{F(s)}_{L^2} \norm{\tilde{K}_2(t)}_{L^\infty_1L^1_2} \\
	&\quad + \norm{F(s)}_{L^1} \norm{K_2(t)-\tilde{K}_2(t)}_{L^\infty_1L^2_2}\Bigr)\diff s \\
	&= \int_0^\infty\ee^{-2Ls}(8\pi s)^{-1/2} \norm{\tilde{K}_2(t)}_{L^\infty_1L^1_2} \diff s + \int_0^\infty\ee^{-2Ls} \norm{K_2(t)-\tilde{K}_2(t)}_{L^\infty_1L^2_2}\diff s \\
	&\leq C (\norm{\tilde{K}_2(t)}_{L^\infty_1L^1_2} + \norm{K_2(t)-\tilde{K}_2(t)}_{L^\infty_1L^2_2}). 
\end{align*}
The term $\norm{K_2(t)-\tilde{K}_2(t)}_{L^\infty_1L^2_2}$ can be bounded by some $C_t$ similarly to \eqref{eq:K2tildeK2}. To bound $\norm{\tilde{K}_2(t)}_{L^\infty_1L^1_2}$, we calculate
\begin{equation*}
	\norm{\tilde{K}_2(t,x,\cdot)}_{L^1} = \int_{\R^2} \abs{u_N(t,x)} \abs{w_N(x-y)} \abs{u_N(t,y)} \diff y \leq \norm{u_N(t)}_{L^\infty}^2 \norm{w_N}_{L^1}\leq C_t. 
\end{equation*}
This completes the proof of \eqref{eq:kperp:infty2bound}.
\end{proof}

\begin{proof}[Proof of Lemma~\ref{lem:bogdyn}]
The well-posedness in Lemma~\ref{lem:bogdyn} follows from a straightforward extension of \cite[Theorem~7]{LewNamSch-15} using the properties \ref{assu:Vd} and \ref{assu:Ve} of $V$, the bound $w^2\leq C(1-\Delta)$, and the bounds on $u_N$ from Lemma~\ref{lem:Hregularity}.

Estimates \eqref{eq:kinetic:bog:poly} and \eqref{eq:kinetic:bog:exp} in \ref{item:kinetic:bog} follow from Lemma~\ref{lem:bogbounds} and a Gr{\"o}nwall argument similar to the proof of \cite[Lemma~5]{NamNap-19}. 

Finally, let us come to the proof of \ref{item:number:bog}. To prove \eqref{eq:number:bog}, we want to use a Gr{\"o}nwall argument. Fix $T>0$ and let $0\leq t\leq T$. We compute
\begin{align*}
	\frac{\diff}{\diff t}\langle \Phi(t),\numop\Phi(t)\rangle &= -\Im \langle\Phi(t),[\BH(t),\numop]\Phi(t)\rangle \\
	&= 2\Im\langle\Phi(t),\iint K_2(t,x,y)a^\ast_xa^\ast_y \diff x\diff y \Phi(t)\rangle. 
\end{align*}
The naive bound
\begin{align*}
	\MoveEqLeft[3] \Im\langle\Phi(t),\iint K_2(t,x,y)a^\ast_xa^\ast_y \diff x\diff y \Phi(t)\rangle \\ 
	&\leq \norm{(1-\Delta_x)^{-1/2}K_2(t,\cdot,\cdot)}_{L^2}\langle\Phi(t),(\numop+1)\Phi(t)\rangle^{1/2}\langle\Phi(t),\sq(1-\Delta) \Phi(t)\rangle^{1/2} \\
	&\leq \langle\Phi(t),(\numop+1)\Phi(t)\rangle + \norm{(1-\Delta_x)^{-1/2}K_2(t,\cdot,\cdot)}_{L^2}^2\langle\Phi(t),\sq(1-\Delta) \Phi(t)\rangle
\end{align*}
is not sufficient and does not take into account that we only want to bound the imaginary part. Hence, we need to replace the $K_2$-term
\[
	\iint K_2(t,x,y)a^\ast_xa^\ast_y \diff x\diff y
\]
by something we can control better. For this purpose, we define 
\begin{equation}\label{eq:defn:k}
	k(t,x,y) = (2L-\Delta_x-\Delta_y)^{-1}K_2(t,x,y)
\end{equation}
for some $L\coloneqq L_T\geq1$ to be determined later. From \eqref{eq:deltaxK2bound} in Lemma~\ref{lem:K2bounds}, we obtain the bounds
\begin{equation}\label{eq:k:L2:bound}
	\norm{k(t)}_{L^2} \leq C_t L^{-1/2}, \qquad \norm{\partial_tk(t)}_{L^2} \leq C_t L^{-1/2}\leq C_t. 
\end{equation}
The proof consists of two steps: In the first step, we replace the $K_2$-term. In the second step, we close the Gr{\"o}nwall argument. 

\paragraph{Step 1: Replacing the $K_2$-term.}

To find a different expression for the $K_2$-term, we want to compute the commutator
\[
	[\BH(t), \iint k(t,x,y)a^\ast_xa^\ast_y \diff x\diff y]. 
\]
We write
\[
	\BH(t) = \sq(L-\Delta) + \sq(A) + \frac{1}{2}\biggl(\iint K_2(t,x,y)a^\ast_xa^\ast_y \diff x\diff y + \hc \biggr),
\]
where $A = A_1 + A_2$ with $A_1 = h(t)+K_1(t)-(L-\Delta+V_t)$ and $A_2 = V_t$. 
We have 
\begin{equation}\label{eq:kinetic:commutator}
\begin{aligned}
	\MoveEqLeft[3] [\sq(L-\Delta), \iint k(t,x,y)a^\ast_xa^\ast_y \diff x\diff y] \\
	&= \iint \bigl( (L-\Delta_x) + (L-\Delta_y)\bigr)k(t,x,y)a^\ast_xa^\ast_y \diff x\diff y \\
	&= \iint K_2(t,x,y)a^\ast_xa^\ast_y \diff x\diff y
\end{aligned}
\end{equation}
and similarly, since $k(t,x,y)$ is symmetric in $x$ and $y$,
\begin{equation}\label{eq:A:commutator}
	[\sq(A), \iint k(t,x,y)a^\ast_xa^\ast_y \diff x\diff y] = 2\iint (A_xk)(t,x,y)a^\ast_xa^\ast_y \diff x\diff y. 
\end{equation}
It is well known that for an operator of the form 
\[
	\mathbb{K}_2 = \iint K(x,y)a^\ast_xa^\ast_y \diff x\diff y
\]
we have the bound
\begin{equation}\label{eq:offd:bound}
\abs{\langle\xi',\mathbb{K}_2\xi\rangle} \leq C\norm{K}_{L^2}\langle\xi',(\numop+1)\xi'\rangle^{1/2} \langle\xi,(\numop+1)\xi\rangle^{1/2}
\end{equation}
for $\xi',\xi\in\mathcal{F}_\perp(t)$. 
Hence, we can estimate for $\xi\in\mathcal{F}_\perp(t)$
\begin{equation*}
	\abs{\langle\xi,\iint (A_xk)(t,x,y)a^\ast_xa^\ast_y \diff x\diff y \xi\rangle} \leq C\norm{(A_xk)(t)}_{L^2} \langle \xi,(\numop+1)\xi\rangle. 
\end{equation*}
Using \eqref{eq:k:L2:bound}, it is easy to check that
\begin{align*}
	\norm{(A_{1,x}k)(t)}_{L^2} &\leq \norm{A_1}_\textnormal{op}\norm{k(t)}_{L^2} \\
	&\leq \bigl(\norm{w_N\ast\abs{u_N(t)}^2}_{L^\infty} + \abs{\mu_N(t)} + \norm{K_1(t)}_{\textnormal{op}} + L\bigr) \norm{k(t)}_{L^2} \leq C_tL^{1/2}. 
\end{align*}
It remains to treat $A_2 = V_t$. By Assumption~\ref{assu:V} \ref{assu:Vc}, $\abs{V(t,x)}\leq C(\abs{x}^2+1)$. Hence, we have
\[
	\norm{V_t(x)k(t,x,y)}_{L^2} \leq C(\norm{\abs{x}^2k(t,x,y)}_{L^2} + \norm{k(t)}_{L^2}). 
\]
To bound $\norm{\abs{x}^2k(t,x,y)}_{L^2}$, we compute 
\begin{align*}
	\MoveEqLeft[3] \abs{x}^2k(t,x,y)=\abs{x}^2(2L-\Delta_x-\Delta_y)^{-1}K_2(t,x,y) \\
	&=4(2L-\Delta_x-\Delta_y)^{-2}K_2(t,x,y) + 8\Delta_x(2L-\Delta_x-\Delta_y)^{-3}K_2(t,x,y) \\
	&\quad - 4\nabla_x\cdot(2L-\Delta_x-\Delta_y)^{-2}xK_2(t,x,y) + (2L-\Delta_x-\Delta_y)^{-1}\abs{x}^2K_2(t,x,y) \\
	&\eqqcolon k_1(t,x,y) + k_2(t,x,y) + k_3(t,x,y) + k_4(t,x,y), 
\end{align*}
where we repeatedly used
\begin{equation*}
	[x_j,(2L-\Delta_x-\Delta_y)^{-1}] 
	= -2\partial_{x_j}(2L-\Delta_x-\Delta_y)^{-2}
\end{equation*}
and $[x_j,\partial_{x_j}]=-1$, $j=1,2$. From \eqref{eq:deltaxK2bound} in Lemma~\ref{lem:K2bounds}, we obtain $\norm{k_j(t)}_{L^2}\leq C_t$, $j=1,2,3,4$.
Thus, we have $\norm{V_t(x)k(t,x,y)}_{L^2}\leq C_t$. In conclusion,
\[
	\norm{(A_{x}k)(t)}_{L^2} \leq C_t L^{1/2}. 
\]
Using the fact that $K_2(t)$ and $k(t)$ are symmetric, we obtain
\begin{equation}\label{eq:K2:commutator}
\begin{aligned}
\MoveEqLeft[3] \biggl[\biggl(\iint K_2(t,x,y)a^\ast_xa^\ast_y \diff x\diff y + \hc\biggr),\iint k(t,x',y')a^\ast_{x'}a^\ast_{y'} \diff x' \diff y'\biggr] \\
	&=\biggl[\iint \overline{K_2(t,x,y)}a_xa_y \diff x\diff y, \iint k(t,x',y')a^\ast_{x'}a^\ast_{y'} \diff x' \diff y'\biggr] \\
	&= 2\langle K_2(t),k(t)\rangle + 4\sq(B(t)),
\end{aligned}
\end{equation}
where $B(t)$ is the operator with kernel 
\begin{equation}\label{eq:defn:B}
	B(t,x,y)=\int k(t,x,x')\overline{K_2(t,x',y)} \diff x'. 
\end{equation}
Collecting \eqref{eq:kinetic:commutator}, \eqref{eq:A:commutator}, and \eqref{eq:K2:commutator}, we obtain
\begin{multline*}
	[\BH(t), \iint k(t,x,y)a^\ast_xa^\ast_y \diff x\diff y] = \iint K_2(t,x,y)a^\ast_xa^\ast_y \diff x\diff y \\
	+ 2\iint (A_xk)(t,x,y)a^\ast_xa^\ast_y \diff x\diff y + \langle K_2(t),k(t)\rangle + 2\sq(B(t)).
\end{multline*}
Therefore, we have
\begin{align*}
	\MoveEqLeft[3] \frac{\diff}{\diff t}\langle \Phi(t),\numop\Phi(t)\rangle=2\Im\langle\Phi(t),\iint K_2(t,x,y)a^\ast_xa^\ast_y \diff x\diff y \Phi(t)\rangle \\
	&=2\Im\langle\Phi(t),[\BH(t), \iint k(t,x,y)a^\ast_xa^\ast_y \diff x\diff y] \Phi(t)\rangle \\
	&\quad - 4\Im\langle\Phi(t),\iint (A_xk)(t,x,y)a^\ast_xa^\ast_y \diff x\diff y \Phi(t)\rangle - 4\Im\langle\Phi(t),\sq(B(t))\Phi(t)\rangle \\
	&=-2\frac{\diff}{\diff t}\Re\langle\Phi(t),\iint k(t,x,y)a^\ast_xa^\ast_y \diff x\diff y \Phi(t)\rangle \\
	&\quad + 2\Re\langle\Phi(t),\iint \partial_tk(t,x,y)a^\ast_xa^\ast_y \diff x\diff y \Phi(t)\rangle \\
	&\quad - 4\Im\langle\Phi(t),\iint (A_xk)(t,x,y)a^\ast_xa^\ast_y \diff x\diff y \Phi(t)\rangle - 4\Im\langle\Phi(t),\sq(B(t))\Phi(t)\rangle,
\end{align*}
where we used that $\langle K_2(t),k(t)\rangle\in\R$ since
\[
	\langle K_2(t),k(t)\rangle = \langle K_2(t),(2L-\Delta_x-\Delta_y)^{-1}K_2(t)\rangle \geq 0. 
\]

\paragraph{Step 2: Closing the Gr{\"o}nwall argument.}

Using \eqref{eq:offd:bound} and $\abs{\langle\xi,\sq(B(t))\xi\rangle}\leq\norm{B(t)}_\textnormal{op}\langle\xi,\numop\xi\rangle$ for $\xi\in\mathcal{F}_\perp(t)$, we find
\begin{align*}
	\MoveEqLeft[3] \frac{\diff}{\diff t}\biggl( \langle \Phi(t),\numop\Phi(t)\rangle + 2\Re\langle\Phi(t),\iint k(t,x,y)a^\ast_xa^\ast_y \diff x\diff y \Phi(t)\rangle \biggr) \\
	&\leq C(\norm{\partial_tk(t)}_{L^2} + \norm{(A_xk)(t)}_{L^2} + \norm{B(t)}_{\textnormal{op}})\langle \Phi(t),(\numop+1)\Phi(t)\rangle \\
	&\leq C_tL^{1/2}\langle \Phi(t),(\numop+1)\Phi(t)\rangle, 
\end{align*}
where we used the bounds $\norm{\partial_tk(t)}_{L^2} \leq C_t$ and $\norm{(A_xk)(t)}_{L^2}\leq C_tL^{1/2}$, which we proved before, and the simple bound $\norm{B(t)}_{\textnormal{op}}\leq\norm{K_2(t)}_{\textnormal{op}}\norm{k(t)}_{\textnormal{op}}\leq C_t$. 
Integrating this inequality and using \eqref{eq:thm:norm:initial}, \eqref{eq:offd:bound}, and \eqref{eq:k:L2:bound}, we obtain
\begin{align*}
	\langle \Phi(t),\numop\Phi(t)\rangle &\leq - 2\Re\langle\Phi(t),\iint k(t,x,y)a^\ast_xa^\ast_y \diff x\diff y \Phi(t)\rangle \\
	&\quad + \langle \Phi(0),\numop\Phi(0)\rangle + 2\Re\langle\Phi(0),\iint k(0,x,y)a^\ast_xa^\ast_y \diff x\diff y \Phi(0)\rangle \\
	&\quad + \int_0^tC_{t'}L^{1/2}\langle \Phi(t'),(\numop+1)\Phi(t')\rangle\diff t' \\
	&\leq C_TL^{-1/2}\langle \Phi(t),(\numop+1)\Phi(t)\rangle + C_t + \int_0^tC_{t'}L^{1/2}\langle \Phi(t'),(\numop+1)\Phi(t')\rangle\diff t',
\end{align*}
where $C_T=\max_{0\leq t\leq T}C_t$. Choosing $L=4C_T^2$, this implies
\[
	\langle \Phi(t),\numop\Phi(t)\rangle \leq C_T + C_T \int_0^t\langle \Phi(t'),\numop\Phi(t')\rangle\diff t'. 
\] 
Using Gr{\"o}nwall's lemma, we obtain
\[
	\langle \Phi(t),\numop\Phi(t)\rangle \leq C_t
\]
for $0\leq t\leq T$. Since $T>0$ was chosen arbitrarily, this completes the proof of \eqref{eq:number:bog}. 
\end{proof}

\section{Bogoliubov's approximation}\label{sec:app}

Following the strategy in \cite{LewNamSch-15,LewNamSerSol-15}, we use the unitary map $U_N(t)\colon\mathfrak{H}^N\to\mathcal{F}_\perp^{\leq N}(t)$ defined by
\[
	U_N(t)=\bigoplus_{k=0}^N q(t)^{\otimes k}\biggl(\frac{a(u_N(t))^{N-k}}{\sqrt{(N-k)!}}\biggr),
\]
where $q(t)=\id-\ketbra{u_N(t)}{u_N(t)}$, to factor out the condensate $u_N(t)$ defined by the Hartree equation \eqref{eq:Hartree}. 
The action of $U_N(t)$ can be computed explicitly using the following identities for $f,g\in\mathfrak{H}_\perp$, which are proved in \cite{LewNamSerSol-15}:
\begin{equation}\label{eq:Uid}
\begin{aligned}
	U_N(t)a^\ast(u_N(t))a(u_N(t))U_N(t)^\ast &=N-\numop, \\
	U_N(t)a^\ast(f)a(u_N(t))U_N(t)^\ast &=a^\ast(f)\sqrt{N-\numop}, \\
	U_N(t)a^\ast(u_N(t))a(g)U_N(t)^\ast &=\sqrt{N-\numop}a(g), \\
	U_N(t)a^\ast(f)a(g)U_N(t)^\ast &=a^\ast(f)a(g).
\end{aligned}
\end{equation}
In the following, we consider the transformed dynamics 
\begin{equation}\label{eq:trafodyn}
	\Phi_N(t)=U_N(t)\Psi_N(t).
\end{equation}
The Schr{\"o}dinger equation \eqref{eq:linearschr} can be written equivalently as
\[
	\begin{cases}
		\ii\partial_t\Phi_N(t)=\mathcal{G}_N(t)\Phi_N(t), \\
		\Phi_N(0)=U_N(0)\Psi_{N,0},
	\end{cases}
\]
where the generator $\mathcal{G}_N(t)$ can be computed explicitly using the identities \eqref{eq:Uid} and the second-quantized form of $H_{N,t}$ \eqref{eq:Hsq}.
For $m\in\N$, we denote the projection onto the truncated Fock space $\mathcal{F}_\perp^{\leq m}(t)$ by $\ind^{\leq m}$. 
Similarly to the calculation in \cite[Appendix~B]{LewNamSch-15}, one can check that $\mathcal{G}_N(t)$ satisfies
\[
	\mathcal{G}_N(t)=(\ii\partial_t U_N(t))U_N(t)^\ast +U_N(t)H_{N,t}U_N(t)^\ast =\ind^{\leq N}\bigl(\BH(t)+\mathcal{E}_N(t)\bigr)\ind^{\leq N},
\]
where 
\[
	\mathcal{E}_N(t)=\frac12\sum_{j=0}^4(R_j+R_j^\ast),
\]
\begin{align}
R_{0}&=R_0^\ast= \sq\bigl(q(t)[w_N\ast\abs{u_N(t)}^2+ \tilde{K}_1(t) -\mu_N(t)]q(t)\bigr)\frac{1-\numop}{N-1}, \notag \\
R_{1}&=-2\frac{\numop\sqrt{N-\numop}}{N-1} a(q(t)[w_N\ast\abs{u_N(t)}^2]u_N(t)), \notag \\
R_{2}&= \iint K_2(t,x,y) a^\ast_x a^\ast_y \diff x \diff y \biggl(\frac{\sqrt{(N-\numop)(N-\numop-1)}}{N-1}-1\biggr), \notag \\
R_{3}& =2\frac{\sqrt{N-\numop}}{N-1}\iiiint( \id \otimes q(t) w_N q(t)\otimes q(t))(x,y;x',y')\overline{u_N(t,x)} a^\ast_y a_{x'} a_{y'} \diff x \diff y \diff x' \diff y', \notag \\
R_{4}&=R_4^\ast= \frac{1}{N-1}\sq_2(q(t)\otimes q(t)w_N q(t)\otimes q(t)). \label{eq:defn:Rj}
\end{align}

To prove bounds on the error term $\mathcal{E}_N(t)$, we need the following operator estimate on $\abs{w_N(x-y)}$, which is an improvement of the bound
\[
	\abs{w_N(x-y)}\leq C_s\norm{w_N}_{L^s}(1-\Delta_x), \qquad s>1,
\]
which follows from the standard Sobolev inequality. 
\begin{lem}\label{lem:MT}
Let $w_N$ satisfy Assumption~\ref{assu:wN}. Then, we have the following operator estimate on $L^2((\R^2)^2)$
\begin{equation}\label{eq:wN:MTineq}
	\abs{w_N(x-y)}\leq C\beta\log(N)(1-\Delta_x)
\end{equation}
if $w_N$ satisfies \eqref{eq:poly} for $\beta>0$. If $w_N$ satisfies \eqref{eq:exp} for $0<\kappa<1$, $\beta\log(N)$ is replaced by $N^\kappa$. 
\end{lem}

\begin{proof}
We focus on the case where $w_N$ satisfies \eqref{eq:poly}, since the case where $w_N$ satisfies \eqref{eq:exp} follows by setting $\beta=N^\kappa/\log(N)$. 

We will use the inequality
\begin{equation}\label{eq:legendre}
	ab\leq (\ee^a-1)+(1-b+b\log(b)), \qquad a,b>0, 
\end{equation}
which follows from computing the Legendre transformation of $a\mapsto\ee^a-1$. 
Recall the Moser--Trudinger inequality (see e.g.\ \cite[Theorem~1.67]{BahCheDan-11})
\begin{equation}\label{eq:MT}
	\int_{\R^2}\Biggl(\exp\biggl(\gamma\biggl(\frac{\abs{g(x)}}{\norm{g}_{H^1}}\biggr)^2\biggr)-1\Biggr)\diff x\leq K, \qquad g\in H^1(\R^2),
\end{equation}
where $\gamma, K>0$ are some fixed constants. 
Let $f\in H^1((\R^2)^2)$ and define
\[
	\tilde{f}_y(x)=\frac{f(x,y)}{\norm{f(\cdot,y)}_{H^1(\R^2)}}. 
\]
We can write
\begin{equation*}
	\iint\abs{w_N(x-y)}\abs{f(x,y)}^2\diff x\diff y=\int\norm{f(\cdot,y)}_{H^1}^2\biggl(\int \abs{w_N(x-y)}\abs{\tilde{f}_y(x)}^2\diff x\biggr)\diff y.
\end{equation*}
For fixed $y$, we denote $\mathcal{L}_y=\{x:\abs{w_N(x-y)}<\gamma\}\subset\R^2$ and $\mathcal{H}_y=\{x:\abs{w_N(x-y)}\geq\gamma\}\subset\R^2$. Using \eqref{eq:legendre} and \eqref{eq:MT}, we compute
\begin{align*}
	\MoveEqLeft[3] \int\abs{w_N(x-y)}\abs{\tilde{f}_y(x)}^2 \diff x =
	\int_{\mathcal{L}_y} \abs{w_N(x-y)}\abs{\tilde{f}_y(x)}^2 \diff x + \int_{\mathcal{H}_y} \abs{w_N(x-y)}\abs{\tilde{f}_y(x)}^2 \diff x \\
	&\leq \gamma\int_{\mathcal{L}_y} \abs{\tilde{f}_y(x)}^2 \diff x + \int_{\mathcal{H}_y} (\ee^{\gamma\abs{\tilde{f}_y(x)}^2}-1) \diff x \\
	&\quad + \int_{\mathcal{H}_y} \bigl(1-\gamma^{-1}\abs{w_N(x-y)}+\gamma^{-1}\abs{w_N(x-y)}\log(\gamma^{-1}\abs{w_N(x-y)})\bigr)\diff x \\
	&\leq \gamma\int \abs{\tilde{f}_y(x)}^2 \diff x + K + \int_{\mathcal{H}_0} \bigl(\gamma^{-1}\abs{w_N(x)}\log(\gamma^{-1}\abs{w_N(x)})\bigr)\diff x \\
	&\leq \gamma\int \abs{\tilde{f}_y(x)}^2 \diff x + C + C\int_{\mathcal{H}_0} N^{2\beta}\abs{w(N^\beta x)}\bigl(\log(N^{2\beta})+\log(\abs{w(N^\beta x)})\bigr)\diff x \\
	&\leq \gamma\int \abs{\tilde{f}_y(x)}^2 \diff x + C + C\beta\log(N)\int_{\R^2}\abs{w(x)}\diff x+C\int_{\R^2}\abs{w(x)}^2\diff x\\
	&\leq \gamma\int \abs{\tilde{f}_y(x)}^2 \diff x + C\beta\log(N).
\end{align*}
Here we used $\abs{w}\log(\abs{w})\leq\abs{w}^2$ and $w\in L^1(\R^2)\cap L^2(\R^2)$. Therefore, we obtain
\begin{align*}
	\iint\abs{w_N(x-y)}\abs{f(x,y)}^2\diff x\diff y &\leq \int\norm{f(\cdot,y)}_{H^1}^2\biggl(\gamma\int \abs{\tilde{f}_y(x)}^2 \diff x + C\beta\log(N)\biggr)\diff y \\
	&\leq \gamma\norm{f}_{L^2}^2 + C\beta\log(N)\norm{f}_{L^2_yH^1_x}^2. 
\end{align*}
This implies \eqref{eq:wN:MTineq}. 
\end{proof}

The following lemma is an improvement of \cite[Lemma~9]{NamNap-19} and justifies why the error terms are small (on sectors of few particles). 
\begin{lem}\label{lem:error-bounds}
For all $1\leq m\leq N$ and $\xi',\xi\in\mathcal{F}_\perp(t)$, we have 
\begin{align}
	\abs{\langle\xi',\ind^{\leq m}R_4\ind^{\leq m}\xi\rangle} &\leq \frac{C\beta\log(N)m}{N}\norm{\sq(1-\Delta)^{1/2}\xi'}\norm{\sq(1-\Delta)^{1/2}\xi}, \label{eq:R4bound} \\
	\abs{\langle\xi',\ind^{\leq m}(\mathcal{E}_N(t)-R_4)\ind^{\leq m}\xi\rangle} &\leq C_t\sqrt{\frac{\beta\log(N)m}{N}}\Bigl(\norm{(\numop+1)^{1/2}\xi'}\norm{\sq(1-\Delta)^{1/2}\xi} \notag \\
	&\mathmakebox[\widthof{$\leq C_t\displaystyle\sqrt{\frac{\beta\log(N)m}{N}}\Bigr(\;$}][r]{+\:} \norm{\sq(1-\Delta)^{1/2}\xi'}\norm{(\numop+1)^{1/2}\xi}\Bigr), \label{eq:eb1} \\ 	
	\abs{\langle\xi',\ind^{\leq m}\partial_t\mathcal{E}_N(t)\ind^{\leq m}\xi\rangle} &\leq C_t\sqrt{\frac{\beta\log(N)m}{N}}\Bigl(\norm{(\numop+1)^{1/2}\xi'}\norm{\sq(1-\Delta)^{1/2}\xi} \notag \\
	&\mathmakebox[\widthof{$\leq C_t\displaystyle\sqrt{\frac{\beta\log(N)m}{N}}\Bigr(\;$}][r]{+\:} \norm{\sq(1-\Delta)^{1/2}\xi'}\norm{(\numop+1)^{1/2}\xi}\Bigr), \label{eq:eb2} \\
	\abs{\langle\xi',\ind^{\leq m}\ii[\mathcal{E}_N(t),\numop]\ind^{\leq m}\xi\rangle} &\leq C_t\sqrt{\frac{\beta\log(N)m}{N}}\Bigl(\norm{(\numop+1)^{1/2}\xi'}\norm{\sq(1-\Delta)^{1/2}\xi} \notag \\
	&\mathmakebox[\widthof{$\leq C_t\displaystyle\sqrt{\frac{\beta\log(N)m}{N}}\Bigr(\;$}][r]{+\:} \norm{\sq(1-\Delta)^{1/2}\xi'}\norm{(\numop+1)^{1/2}\xi}\Bigr) \label{eq:eb3}
\end{align}
if $w_N$ satisfies \eqref{eq:poly} for $\beta>0$. If $w_N$ satisfies \eqref{eq:exp} for $0<\kappa<1$, $\beta\log(N)$ is replaced by $N^\kappa$. 
\end{lem}

\begin{proof}
The estimates \eqref{eq:R4bound}, \eqref{eq:eb1}, \eqref{eq:eb2}, and \eqref{eq:eb3} follow from the proof of \cite[Lemma~9]{NamNap-19} by using the operator estimate \eqref{eq:wN:MTineq}. 
\end{proof}

The next lemma is an upper bound on the kinetic energy of $\Phi_N(t)$. Note that we generally do not have stability of second kind for our Hamiltonian $H_{N,t}$ unless $\beta\leq1/2$; see \cite[Lemma~2]{LewNamRou-17}. To circumvent this problem, we use the fact that our initial data are localized and cannot spread too fast, which enables us to use stability of second kind for trapped Hamiltonians \cite{JunVis-26} to bound the kinetic energy. 
\begin{lem}\label{lem:apkinetic}
Under the same assumptions as in Theorem~\ref{thm:norm}, we have
\[
	\langle\Phi_N(t),\sq(1-\Delta)\Phi_N(t)\rangle\leq C_t N. 
\]
\end{lem}

\begin{proof}
We compute 
\begin{multline*}
	U_N(t)\sq(-\Delta)U_N(t)^\ast = \langle u_N(t),-\Delta u_N(t)\rangle (N-\numop) \\
	+ \bigl( a^\ast(q(t)(-\Delta u_N)(t))\sqrt{N-\numop}+\hc\bigr) + \sq(q(t)(-\Delta)q(t)). 
\end{multline*}
Since $\Phi_N(t)=U_N(t)\Psi_N(t)\in\mathcal{F}_\perp^{\leq N}(t)$, we have
\begin{align*}
	\MoveEqLeft[3] \langle\Phi_N(t),\sq(1-\Delta)\Phi_N(t)\rangle = \langle\Phi_N(t),\numop\Phi_N(t)\rangle + \langle\Phi_N(t),\sq(q(t)(-\Delta)q(t))\Phi_N(t)\rangle \\
	&= \langle\Phi_N(t),\numop\Phi_N(t)\rangle - \langle\Phi_N(t),\bigl(a^\ast(q(t)(-\Delta u_N)(t))\sqrt{N-\numop}+\hc\bigr)\Phi_N(t)\rangle \\
	&\quad - \langle u_N(t),-\Delta u_N(t)\rangle\langle\Phi_N(t),(N-\numop)\Phi_N(t)\rangle + \langle \Psi_N(t),\sq(-\Delta)\Psi_N(t)\rangle. 
\end{align*}
Using
\[
	a^\ast(g)a(g)\leq\norm{g}_{L^2}^2\numop, \qquad g\in L^2(\R^2)
\]
on $\mathcal{F}_\perp(t)$ and $\norm{u_N(t)}_{H^2}\leq C_t$ from Lemma~\ref{lem:Hregularity}, the first three terms on the right-hand side of the equation can be bounded easily by $C_t N$. In the remainder of the proof, we bound the term $\langle \Psi_N(t),\sq(-\Delta)\Psi_N(t)\rangle$. 

Let us first find a bound on
\[
	\langle\Psi_N(t),\sum_{j=1}^N\abs{x_j}^s\Psi_N(t)\rangle \leq \langle\Psi_N(t),\sum_{j=1}^N\jbr{x_j}^s\Psi_N(t)\rangle
\]
for $0<s\leq1$. 
We compute
\begin{align*}
	\MoveEqLeft[3] \abs[\bigg]{\frac{\diff}{\diff t}\langle\Psi_N(t),\sum_{j=1}^N\jbr{x_j}^s\Psi_N(t)\rangle}
	= \abs[\bigg]{\langle\Psi_N(t),[H_{N,t},\sum_{j=1}^N\jbr{x_j}^s]\Psi_N(t)\rangle} \\
	&= \abs[\bigg]{\langle\Psi_N(t),\sum_{j=1}^N[-\Delta_j,\jbr{x_j}^s]\Psi_N(t)\rangle} \\
	&= \abs[\bigg]{\sum_{j=1}^N\langle\Psi_N(t),\Biggl(-\biggl(\frac{2s-s^2}{\jbr{x_j}^{4-s}}+\frac{s^2}{\jbr{x_j}^{2-s}}\biggr)-\frac{2s}{\jbr{x_j}^{2-s}}x_j\cdot\nabla_j\Biggr)\Psi_N(t)\rangle} \\
	&\leq CN+ \sum_{j=1}^N\langle\Psi_N(t),-\Delta_j\Psi_N(t)\rangle = CN+\langle \Psi_N(t),\sq(-\Delta)\Psi_N(t)\rangle. 
\end{align*}
Using \eqref{eq:thm:norm:initial-ext}, this implies
\begin{equation}\label{eq:external-bound}
	\langle\Psi_N(t),\sum_{j=1}^N\jbr{x_j}^s\Psi_N(t)\rangle \leq C_tN + \int_0^t\langle \Psi_N(t'),\sq(-\Delta)\Psi_N(t')\rangle\diff t'. 
\end{equation}
Let $0<\varepsilon<1$ be such that 
\[
	(1-\varepsilon)^{-1}\int_{\R^2}w_-<a^\ast.
\]
Since the modified external potential $(1-\varepsilon)^{-1}(V(t,x)+\abs{x}^s)$ satisfies 
\[
	(1-\varepsilon)^{-1}(V(t,x)+\abs{x}^s)\geq (1-\varepsilon)^{-1}\abs{x}^s-(1-\varepsilon)^{-1}c_t
\]
by Assumption~\ref{assu:V} \ref{assu:Vd} on $V$, \cite[Theorem~2]{JunVis-26} guarantees that the modified Hamiltonian
\begin{multline*}
	\sum_{j=1}^N\bigl(-\Delta_j+(1-\varepsilon)^{-1}(V(t,x_j)+\abs{x_j}^s)\bigr)+\frac{1}{N-1}\sum_{1\leq j<k\leq N}(1-\varepsilon)^{-1}w_N(x_j-x_k) \\
	= (1-\varepsilon)^{-1}\biggl(H_{N,t}+\sum_{j=1}^N(\varepsilon\Delta_j+\abs{x_j}^s)\biggr)
\end{multline*}
is stable of second kind, i.e.,
\begin{equation}\label{eq:HNt-lowerbound}
	H_{N,t}+\sum_{j=1}^N(\varepsilon\Delta_j+\abs{x_j}^s) \geq -C_tN.
\end{equation}
With this bound and \eqref{eq:external-bound}, we obtain
\begin{multline}\label{eq:PsiN-kinetic-bound-1}
	\varepsilon\langle\Psi_N(t),\sq(-\Delta)\Psi_N(t)\rangle \leq
	C_t N + \langle\Psi_N(t),H_{N,t}\Psi_N(t)\rangle + \langle\Psi_N(t),\sum_{j=1}^N\abs{x_j}^s\Psi_N(t)\rangle \\
	\leq C_t N + \langle\Psi_N(t),H_{N,t}\Psi_N(t)\rangle + \int_0^t\langle \Psi_N(t'),\sq(-\Delta)\Psi_N(t')\rangle\diff t'. 
\end{multline}
By assumptions \ref{assu:Vd} and \ref{assu:Ve} on $V_t$, we have
\begin{equation*}
	\abs[\bigg]{\frac{\diff}{\diff t}\langle\Psi_N(t),H_{N,t}\Psi_N(t)\rangle} \leq \abs{\langle\Psi_N(t),\sq(\partial_tV_t)\Psi_N(t)\rangle}
	\leq C_t \langle\Psi_N(t),\sq(V_t+C_t)\Psi_N(t)\rangle. 
\end{equation*}
Similarly to \eqref{eq:HNt-lowerbound}, the condition
\[
	\int_{\R^2}w_-<a^\ast
\]
implies by \cite[Theorem~2]{JunVis-26} that the Hamiltonian
\[
	\sum_{j=1}^N\bigl(-\Delta_j+\abs{x_j}^s\bigr)+\frac{1}{N-1}\sum_{1\leq j<k\leq N}w_N(x_j-x_k) = H_{N,t}+\sum_{j=1}^N(-V(t,x_j)+\abs{x_j}^s)
\]
is stable of second kind, i.e.,
\[
	H_{N,t}+\sum_{j=1}^N(-V(t,x_j)+\abs{x_j}^s) \geq -CN. 
\]
Thus, in combination with \eqref{eq:external-bound}, we obtain
\begin{multline*}
	\abs[\bigg]{\frac{\diff}{\diff t}\langle\Psi_N(t),H_{N,t}\Psi_N(t)\rangle} \leq C_tN + C_t \langle\Psi_N(t),H_{N,t}\Psi_N(t)\rangle + C_t\langle\Psi_N(t),\sum_{j=1}^N\abs{x_j}^s\Psi_N(t)\rangle \\
	\leq C_tN + C_t \langle\Psi_N(t),H_{N,t}\Psi_N(t)\rangle + C_t\int_0^t\langle \Psi_N(t'),\sq(-\Delta)\Psi_N(t')\rangle\diff t'. 
\end{multline*}
Gr{\"o}nwall's lemma implies
\begin{multline*}
	\langle\Psi_N(t),H_{N,t}\Psi_N(t)\rangle \\
	\leq C_t N + C_t \abs{\langle\Psi_N(0),H_{N,0}\Psi_N(0)\rangle} + C_t \int_0^t\langle \Psi_N(t'),\sq(-\Delta)\Psi_N(t')\rangle\diff t'. 
\end{multline*}
Inserting this in \eqref{eq:PsiN-kinetic-bound-1}, we obtain
\begin{multline*}
	\langle\Psi_N(t),\sq(-\Delta)\Psi_N(t)\rangle \\ 
	\leq C_t\biggl(N + \abs{\langle\Psi_N(0),H_{N,0}\Psi_N(0)\rangle} + \int_0^t\langle \Psi_N(t'),\sq(-\Delta)\Psi_N(t')\rangle\diff t'\biggr). 
\end{multline*}
Applying Gr{\"o}nwall's lemma again, we find the desired bound
\[
	\langle\Psi_N(t),\sq(-\Delta)\Psi_N(t)\rangle \leq C_t(N + \abs{\langle\Psi_N(0),H_{N,0}\Psi_N(0)\rangle})\leq C_t N.
\]
Here we used $\abs{\langle\Psi_N(0),H_{N,0}\Psi_N(0)\rangle}\leq CN$, which follows from \eqref{eq:thm:norm:initial} and the bounds from Lemmas~\ref{lem:bogbounds} and \ref{lem:error-bounds} as explained in more detail in the proof of \cite[Lemma~10]{NamNap-19}. 
\end{proof}

\section{From the full to the Bogoliubov dynamics}\label{sec:ftb}

Let $M=N^{1-\delta}$ with $\delta\in(0,1)$. We follow \cite{NamNap-19} and introduce a truncated dynamics $\Phi_{N,M}(t)\in\mathcal{F}_\perp(t)$, but instead of truncating the full generator $\mathcal{G}_N(t)$ with the sharp cut-off $\ind^{\leq M}$, we only truncate the error terms $\mathcal{E}_N(t)$ with a smooth cut-off $\chi_M\approx\ind^{\leq M}$. 
More precisely, let $\chi\colon\R\to[0,1]$ be a smooth function such that $\chi(r)=1$ if $r\leq3/4$ and $\chi(r)=0$ if $r\geq1$, and define the operator $\chi_{M}$ on $\mathcal{F}$ by
\begin{equation*}
	\chi_{M} = \chi_{M}(\numop) = \chi\biggl(\frac{\numop}{M}\biggr). 
\end{equation*}
Let $\Phi_{N,M}$ be the solution to the equation
\begin{equation}\label{eq:truncdyn}
	\begin{cases}
		\ii\partial_t\Phi_{N,M}(t) = \mathcal{G}_{N,M}(t)\Phi_{N,M}(t), \\
		\Phi_{N,M}(0)=\Phi_0,
	\end{cases}
\end{equation}
where 
\[
	 \mathcal{G}_{N,M}(t) = \BH(t) + \chi_M\mathcal{E}_N(t)\chi_M. 
\]
Similarly to the Bogoliubov equation, the well-posedness of this equation can be proved by applying the method used to show \cite[Theorem~7]{LewNamSch-15}. 
Instead of comparing $\Phi_N$ and $\Phi$ directly, we first compare $\Phi_N$ with $\Phi_{N,M_1}$ and $\Phi_{N,M_1}$ with $\Phi_{N,M_2}$ in Lemma~\ref{lem:tt} and then we compare $\Phi_{N,M_2}$ with $\Phi$ in Lemma~\ref{lem:tbog}. We will choose
\[
	M_2\ll M_1\ll N. 
\]
Let us first state a useful bound for the kinetic energy of the truncated dynamics. 
\begin{lem}\label{lem:truncatedkinetic}
Let the initial datum $\Phi_0$ satisfy \eqref{eq:thm:norm:initial} and let $M=N^{1-\delta}$, $\delta\in(0,1)$. If the interaction potential $w_N$ satisfies \eqref{eq:poly} for $\beta>0$, then we have
\begin{equation*}
	\langle\Phi_{N,M}(t),\sq(1-\Delta+v)\Phi_{N,M}(t)\rangle\leq C_{t}\beta\log(N). 
\end{equation*}
If $w_N$ satisfies \eqref{eq:exp} for $0<\kappa<1$ and $\delta>\kappa$, then we have
\begin{equation*}
	\langle\Phi_{N,M}(t),\sq(1-\Delta+v)\Phi_{N,M}(t)\rangle\leq C_{t}N^{\kappa}. 
\end{equation*}
\end{lem}

\begin{proof}
The lemma follows from the arguments in the proof of \cite[Lemma~11]{NamNap-19} by using Lemmas~\ref{lem:bogbounds} and \ref{lem:error-bounds}. The additional requirement $\delta>\kappa$ if $w_N$ satisfies \eqref{eq:exp} is needed to absorb the error terms from Lemma~\ref{lem:error-bounds} in the $\sq(1-\Delta)$ term from Lemma~\ref{lem:bogbounds}. 
\end{proof}

We also have the following bound on the particle number of the truncated dynamics. 
\begin{lem}\label{lem:truncatednumber}
Let the initial datum $\Phi_0$ satisfy \eqref{eq:thm:norm:initial} and let $M=N^{1-\delta}$, $\delta\in(0,1)$. If the interaction potential $w_N$ satisfies \eqref{eq:poly} for $\beta>0$ or \eqref{eq:exp} for $0<\kappa<1/2$ and $\delta\in(2\kappa,1)$, then we have
\begin{equation*}
	\langle\Phi_{N,M}(t),\numop\Phi_{N,M}(t)\rangle\leq C_{t}. 
\end{equation*}
\end{lem}

\begin{proof}
Let us only prove the case where $w_N$ satisfies \eqref{eq:exp} for $0<\kappa<1/2$. The case where $w_N$ satisfies \eqref{eq:poly} follows analogously. 
Fix $T>0$ and let $0\leq t\leq T$. Recall $A=h(t)+K_1(t)-(L-\Delta)$, the definition of $B(t)$ \eqref{eq:defn:B}, the definition of $k$ \eqref{eq:defn:k}, and the bound $\norm{k(t)}_{L^2} \leq C_t L^{-1/2}$ from \eqref{eq:k:L2:bound} for some sufficiently large $L=L_T\geq1$. As in the proof of \eqref{eq:number:bog}, we compute 
\begin{align*}
	\MoveEqLeft[3] \frac{\diff}{\diff t}\langle \Phi_{N,M}(t),\numop\Phi_{N,M}(t)\rangle = -\Im \langle\Phi_{N,M}(t),[\BH(t),\numop]\Phi_{N,M}(t)\rangle \\
	&\quad - \Im \langle\Phi_{N,M}(t),\chi_M[\mathcal{E}_N(t),\numop]\chi_M\Phi_{N,M}(t)\rangle \\
	&=2\Im\langle\Phi_{N,M}(t),[\BH(t), \iint k(t,x,y)a^\ast_xa^\ast_y \diff x\diff y] \Phi_{N,M}(t)\rangle \\
	&\quad - 4\Im\langle\Phi_{N,M}(t),\iint (A_xk)(t,x,y)a^\ast_xa^\ast_y \diff x\diff y \Phi_{N,M}(t)\rangle \\
	&\quad - 4\Im\langle\Phi_{N,M}(t),\sq(B(t))\Phi_{N,M}(t)\rangle \\
	&\quad - \Im \langle\Phi_{N,M}(t),\chi_M[\mathcal{E}_N(t),\numop]\chi_M\Phi_{N,M}(t)\rangle \\
	&=-2\frac{\diff}{\diff t}\Re\langle\Phi_{N,M}(t),\iint k(t,x,y)a^\ast_xa^\ast_y \diff x\diff y \Phi_{N,M}(t)\rangle - I_1 + I_2 - I_3 - I_4 - I_5,
\end{align*}
where
\begin{align*}
	I_1 &= 2\Im\langle\Phi_{N,M}(t),[\chi_M\mathcal{E}_N(t)\chi_M, \iint k(t,x,y)a^\ast_xa^\ast_y \diff x\diff y] \Phi_{N,M}(t)\rangle, \\
	I_2 &= 2\Re\langle\Phi_{N,M}(t),\iint \partial_tk(t,x,y)a^\ast_xa^\ast_y \diff x\diff y \Phi_{N,M}(t)\rangle, \\
	I_3 &= 4\Im\langle\Phi_{N,M}(t),\iint (A_xk)(t,x,y)a^\ast_xa^\ast_y \diff x\diff y \Phi_{N,M}(t)\rangle, \\
	I_4 &= 4\Im\langle\Phi_{N,M}(t),\sq(B(t))\Phi_{N,M}(t)\rangle, \\
	I_5 &= \Im \langle\Phi_{N,M}(t),\chi_M[\mathcal{E}_N(t),\numop]\chi_M\Phi_{N,M}(t)\rangle. 
\end{align*}
In the remainder of the proof, we will bound all $I_j$, $j=1,2,3,4,5$, and then close the Gr{\"o}nwall argument. 
$I_2$, $I_3$, and $I_4$ can be bounded as in the proof of \eqref{eq:number:bog},
\[
	I_2 - I_3 - I_4 \leq C_t \langle\Phi_{N,M}(t),(\numop+1)\Phi_{N,M}(t)\rangle. 
\]
$I_5$ can be bounded using \eqref{eq:eb3} from Lemma~\ref{lem:error-bounds} and Lemma~\ref{lem:truncatedkinetic},
\begin{align*}
	\pm I_5 &\leq C_t\sqrt{\frac{N^\kappa M}{N}}\langle\Phi_{N,M}(t),(\numop+1)\Phi_{N,M}(t)\rangle^{1/2}\langle\Phi_{N,M}(t),\sq(1-\Delta)\Phi_{N,M}(t)\rangle^{1/2} \\
	&\leq C_t\langle\Phi_{N,M}(t),(\numop+1)\Phi_{N,M}(t)\rangle + C_tMN^{2\kappa-1}. 
\end{align*}
It remains to bound $I_1$. We denote
\begin{align*}
	\K & = \iint k(t,x,y)a^\ast_xa^\ast_y \diff x\diff y, \\
	\delta_+ & = \chi_M(\numop+2) - \chi_M(\numop), \\
	\delta_- & = \chi_M(\numop) - \chi_M(\numop-2)
\end{align*}
and decompose 
\begin{multline*}
	[\chi_M\mathcal{E}_N(t)\chi_M, \iint k(t,x,y)a^\ast_xa^\ast_y \diff x\diff y] \\
	= \chi_M [\mathcal{E}_N(t),\K] \chi_M + \chi_M \calENt \K \delta_+ + \delta_- \K \calENt \chi_M. 
\end{multline*}
To bound the term in $I_1$ containing the commutator $[\calENt,\K]$, we write
\[
	[\calENt,\K] = \frac{1}{2} \sum_{j=0}^4 \bigl([R_j,\K] + [R_j^\ast,\K]\bigr), 
\]
where $R_j$, $j=0,\dotsc,4$, are defined in \eqref{eq:defn:Rj}. 
We will now estimate all terms in $I_1$, step by step. 
\paragraph{Estimate of the $[R_0,\K]$ term.} 
We denote
\[
	\omega_0(\numop) = \frac{1-\numop}{N-1}, \qquad W_0 = q(t)[w_N\ast\abs{u_N(t)}^2+ \tilde{K}_1(t) -\mu_N(t)]q(t). 
\]
It is easy to check that $\norm{W_0}_\textnormal{op}\leq C_t$. 
Then, $R_0 = \sq(W_0)\omega_0(\numop)$ and, using \eqref{eq:A:commutator},
\[
	[R_0,\K] = \frac{-2}{N-1}\sq(W_0)\K + 2\biggl(\iint (W_{0,x}k)(t,x,y)a_x^\ast a_y^\ast \diff x \diff y\biggr) \omega_0(\numop). 
\]
Using \eqref{eq:offd:bound} and the bound $\pm\sq(W_0)\leq \norm{W_0}_\textnormal{op}\numop$, we find
\begin{equation}\label{eq:R0K:bound}
\begin{aligned}
	\MoveEqLeft[3] \pm\Im\langle\Phi_{N,M}(t),\chi_M[R_0,\K]\chi_M\Phi_{N,M}(t)\rangle \\
	&\leq \frac{C}{N}\norm{W_0}_\textnormal{op}\norm{k(t)}_{L^2} \langle\Phi_{N,M}(t),\chi_M(\numop+1)^2\chi_M\Phi_{N,M}(t)\rangle \\
	&\leq \frac{C_tM}{N}\langle\Phi_{N,M}(t),(\numop+1)\Phi_{N,M}(t)\rangle,
\end{aligned}
\end{equation}
where we used $\norm{W_0}_\textnormal{op}\leq C_t$ and $\norm{k(t)}_{L^2}\leq C_t$. 

\paragraph{Estimate of the $[R_1,\K]+[R_1^\ast,\K]$ term.} 
We denote
\[
	\omega_1(\numop) = -2\frac{\numop\sqrt{N-\numop}}{N-1}, \qquad  W_1 = q(t)[w_N\ast\abs{u_N(t)}^2]u_N(t). 
\]
It is easy to check that $\norm{W_1}_{L^2}\leq C_t$. 
Then, we have $R_1=\omega_1(\numop)a(W_1)$ and
\[
	[R_1,\K] = 2\omega_1(\numop)\iint \overline{W_{1}(x)}k(t,x,y)a^\ast_y\diff x\diff y + (\omega_1(\numop)-\omega_1(\numop-2))\K a(W_1). 
\]
Using the bounds $\omega_1(\numop)^2\leq CN^{-1}(\numop+1)^2$ and $\pm (\omega_1(\numop)-\omega_1(\numop-2)) \leq CN^{-1/2}$ on $\mathcal{F}_\perp^{\leq M}(t)$ and \eqref{eq:offd:bound}, we have
\begin{align*}
	\MoveEqLeft[3] \pm\Im\langle\Phi_{N,M}(t),\chi_M[R_1,\K]\chi_M\Phi_{N,M}(t)\rangle \\
	&\leq \eta C \iint \abs{W_1(x)}^2\langle\Phi_{N,M}(t),a_y^\ast a_y\Phi_{N,M}(t)\rangle \diff x\diff y \\
	&\quad + \eta^{-1} C \norm{k(t)}_{L^2}^2\langle\Phi_{N,M}(t),\chi_M(\omega_1(\numop+1))^2\chi_M\Phi_{N,M}(t)\rangle \\
	&\quad + CN^{-1/2}\norm{k(t)}_{L^2}\norm{W_1}_{L^2}\langle\Phi_{N,M}(t),\chi_M(\numop+1)^{3/2}\chi_M\Phi_{N,M}(t)\rangle \\
	&\leq C_t\sqrt{\frac{M}{N}} \langle\Phi_{N,M}(t),(\numop+1)\Phi_{N,M}(t)\rangle,
\end{align*}
where we optimized over $\eta>0$ in the last step and used the bounds $\norm{W_1}_{L^2}\leq C_t$ and $\norm{k(t)}_{L^2}\leq C_t$. 
The term containing $[R_1^\ast,\K]$ can be estimated analogously and yields the same bound. In conclusion, we have
\begin{equation}\label{eq:R1K:bound}
	\pm\Im\langle\Phi_{N,M}(t),\chi_M\bigl([R_1,\K] + [R_1^\ast,\K]\bigr)\chi_M\Phi_{N,M}(t)\rangle 
	\leq C_t\sqrt{\frac{M}{N}} \langle\Phi_{N,M}(t),(\numop+1)\Phi_{N,M}(t)\rangle. 
\end{equation}

\paragraph{Estimate of the $[R_2,\K]+[R_2^\ast,\K]$ term.} 
We denote
\[
	\omega_2(\numop) = \frac{\sqrt{(N-\numop)(N-\numop-1)}}{N-1}-1. 
\]
Recall 
\[
	R_2^\ast = \omega_2(\numop) \iint \overline{K_2(t,x,y)}a_xa_y \diff x\diff y.
\]
Using \eqref{eq:K2:commutator}, we have
\begin{align*}
	[R_2^\ast,\K] &= \omega_2(\numop) \biggl[\iint \overline{K_2(t,x,y)}a_xa_y \diff x\diff y, \iint k(t,x',y')a^\ast_{x'}a^\ast_{y'} \diff x' \diff y'\biggr] \\
	&\quad + (\omega_2(\numop)-\omega_2(\numop-2)) \K \iint \overline{K_2(t,x,y)}a_xa_y \diff x\diff y \\
	&= \omega_2(\numop)\bigl(2\langle K_2(t),k(t)\rangle + 4\sq(B(t))\bigr) \\
	&\quad + (\omega_2(\numop)-\omega_2(\numop-2)) \K \iint \overline{K_2(t,x,y)}a_xa_y \diff x\diff y,
\end{align*}
where the kernel of the operator $B(t)$ is defined in \eqref{eq:defn:B}. Recall $\norm{B(t)}_\textnormal{op}\leq C_t$. Using the fact that $\chi_M\omega_2(\numop)\chi_M$ is self-adjoint, $\langle K_2(t),k(t)\rangle\in\R$, the bounds $\pm(\omega_2(\numop)-\omega_2(\numop-2)) \leq CN^{-1}$ and $\pm\omega_2(\numop)\leq CN^{-1}(\numop+1)$ on $\mathcal{F}_\perp^{\leq M}(t)$, and \eqref{eq:offd:bound}, we obtain
\begin{align*}
	\MoveEqLeft[3] \pm\Im\langle\Phi_{N,M}(t),\chi_M[R_2^\ast,\K]\chi_M\Phi_{N,M}(t)\rangle \leq \frac{CM}{N}\norm{B(t)}_\textnormal{op}\langle\Phi_{N,M}(t),\numop\Phi_{N,M}(t)\rangle  \\
	&\quad + \abs[\bigg]{\iint \langle \chi_M\Phi_{N,M}(t), (\omega_2(\numop)-\omega_2(\numop-2)) \K \overline{K_2(t,x,y)}a_{x}a_{y}\chi_M\Phi_{N,M}(t)\rangle \diff x\diff y}  \\
	&\leq \frac{C_tM}{N}\langle\Phi_{N,M}(t),\numop\Phi_{N,M}(t)\rangle  \\
	&\quad + \eta C_tN^{-1}\norm{(1-\Delta_x)^{-1/2}K_2(t,\cdot,\cdot)}_{L^2}^2 \langle\Phi_{N,M}(t),(\numop+1)\Phi_{N,M}(t)\rangle  \\
	&\quad + \eta^{-1} \frac{C_tM^2}{N} \langle\Phi_{N,M}(t),\sq(1-\Delta)\Phi_{N,M}(t)\rangle  \\
	&\leq C_t \langle\Phi_{N,M}(t),(\numop+1)\Phi_{N,M}(t)\rangle + C_tM^2N^{2\kappa-2}, 
\end{align*}
where we used Lemmas~\ref{lem:K2bounds} and \ref{lem:truncatedkinetic} and chose $\eta=N^{1-\kappa}$ in the last step. 
The $[R_2,\K]$ term can be estimated analogously and yields the same bound. In conclusion, we have
\begin{multline}\label{eq:R2K:bound}
	 \pm\Im\langle\Phi_{N,M}(t),\chi_M\bigl([R_2,\K] + [R_2^\ast,\K]\bigr)\chi_M\Phi_{N,M}(t)\rangle \\
	 \leq C_t \langle\Phi_{N,M}(t),(\numop+1)\Phi_{N,M}(t)\rangle + C_tM^2N^{2\kappa-2}. 
\end{multline}

\paragraph{Estimate of the $[R_3,\K]$ term.} 
We denote
\[
	\omega_3(\numop) = 2\frac{\sqrt{N-\numop}}{N-1}, \quad \K_\perp = \iint k_\perp(t,x,y)a^\ast_xa^\ast_y\diff x\diff y, \quad k_\perp(t)=q(t)\otimes q(t)k(t). 
\]
On $\mathcal{F}_\perp^{\leq N}(t)$, we have in the sense of quadratic forms
\[
	R_3 = \omega_3(\numop)\iint w_N(x-y)\overline{u_N(t,x)}a^\ast_ya_xa_y\diff x\diff y
\]
and
\begin{align*}
	[R_3,\K] = [R_3,\K_\perp] &= \omega_3(\numop) \biggl(2 \iint w_N(x-y)\overline{u_N(t,x)} k_\perp(t,x,y)a^\ast_y\diff x\diff y \\
	&\quad + 2 \iiint w_N(x-y)\overline{u_N(t,x)} k_\perp(t,y,x') a^\ast_ya^\ast_{x'}a_x\diff x\diff y\diff x' \\
	&\quad + 2 \iiint w_N(x-y)\overline{u_N(t,x)} k_\perp(t,x,x') a^\ast_ya^\ast_{x'}a_{y}\diff x\diff y\diff x' \biggr) \\
	&\quad + (\omega_3(\numop)-\omega_3(\numop-2)) \K_\perp \iint w_N(x-y)\overline{u_N(t,x)}a^\ast_ya_xa_y\diff x\diff y \\
	&\eqqcolon \omega_3(\numop) (J_1 + J_2 + J_3) + J_4.
\end{align*}
Here, using $\K_\perp$ instead of $\K$ allowed us to omit the projection $q(t)$ in the definition of $R_3$. 
We can bound for $\xi',\xi\in\mathcal{F}_\perp^{\leq N}(t)$
\begin{align*}
	\abs{\langle\xi',J_1\xi\rangle} &\leq C\norm{u_N(t)}_{L^\infty}\biggl(\eta\norm{w_N}_{L^1}\langle\xi',\numop\xi'\rangle \\
	&\quad + \eta^{-1}\iint \abs{w_N(x-y)}\abs{k_\perp(t,x,y)}^2\diff x \diff y\norm{\xi}^2\biggr) \\
			&\leq C\norm{u_N(t)}_{L^\infty}(\eta\langle\xi',\numop\xi'\rangle + \eta^{-1}N^\kappa \norm{(1-\Delta_x)^{1/2}k_\perp(t,\cdot,\cdot)}_{L^2}^2\norm{\xi}^2) \\
			&\leq C_t(\eta\langle\xi',\numop\xi'\rangle + \eta^{-1}N^\kappa (\norm{(1-\Delta_x)^{-1/2}K_2(t,\cdot,\cdot)}_{L^2}^2+C_t)\norm{\xi}^2) \\
			&\leq C_t (N^{1/2}\langle\xi',\numop\xi'\rangle + N^{2\kappa-1/2}\norm{\xi}^2),
\end{align*}
where we used Lemmas~\ref{lem:MT} and \ref{lem:K2bounds} and chose $\eta=N^{1/2}$ in the last step. Using Lemma~\ref{lem:K2bounds} again, we have for $\xi',\xi\in\mathcal{F}_\perp^{\leq M}(t)$
\begin{align*}
	\abs{\langle\xi',J_2\xi\rangle} 
	&\leq C\norm{u_N(t)}_{L^\infty}\bigl(M^{-1/2}\norm{w_N}_{L^1}\langle\xi',\numop^2\xi'\rangle + M^{1/2}\norm{w_N}_{L^1}\norm{k_\perp(t)}_{L^\infty_{1}L^2_{2}}^2 \langle\xi,\numop\xi\rangle\bigr) \\
	&\leq C_t M^{1/2}(\langle\xi',\numop\xi'\rangle + \langle\xi,\numop\xi\rangle).
\end{align*}
Similarly, we have for $\xi',\xi\in\mathcal{F}_\perp^{\leq M}(t)$
\begin{equation*}
	\abs{\langle\xi',J_3\xi\rangle} 
	\leq C_t M^{1/2}(\langle\xi',\numop\xi'\rangle + \langle\xi,\numop\xi\rangle). 
\end{equation*}
Using $\pm\omega_3(\numop)\leq CN^{-1/2}$ and the bounds on $J_1$, $J_2$, and $J_3$, we obtain
\[
	 \pm\Im\langle\Phi_{N,M}(t),\chi_M\omega_3(\numop)(J_1+J_2+J_3)\chi_M\Phi_{N,M}(t)\rangle
	 \leq C_t( \langle\Phi_{N,M}(t),\numop\Phi_{N,M}(t)\rangle + N^{2\kappa -1}). 
\]
Finally, using $\pm(\omega_3(\numop)-\omega_3(\numop-2))\leq CN^{-3/2}$ on $\mathcal{F}_\perp^{\leq M}(t)$, \eqref{eq:offd:bound}, and Lemmas~\ref{lem:MT} and \ref{lem:truncatedkinetic}, we find 
\begin{align*}
	\MoveEqLeft[3] \pm\Im\langle\Phi_{N,M}(t),\chi_MJ_4\chi_M\Phi_{N,M}(t)\rangle \\
	&\leq \iint \abs{\langle \chi_M\Phi_{N,M}(t), (\omega_3(\numop)-\omega_3(\numop-2)) \K_\perp \times \\
	&\quad \times w_N(x-y)\overline{u_N(t,x)}a^\ast_{y}a_{x}a_{y}\chi_M\Phi_{N,M}(t)\rangle}\diff x\diff y \\
	&\leq CMN^{-3/2}\norm{k(t)}_{L^2}\norm{u_N(t)}_{L^\infty} \Bigl( \eta\norm{w_N}_{L^1} \langle\Phi_{N,M}(t),\numop\Phi_{N,M}(t)\rangle \\
	&\quad + \eta^{-1}N^\kappa\iint\langle\Phi_{N,M}(t),\chi_Ma^\ast_ya^\ast_x(1-\Delta_x)a_xa_y\chi_M\Phi_{N,M}(t)\rangle\diff x\diff y \Bigr) \\
	&\leq C_tMN^{-3/2} ( \eta \langle\Phi_{N,M}(t),\numop\Phi_{N,M}(t)\rangle +  \eta^{-1}MN^{2\kappa} ) \\
	&\leq C_t\langle\Phi_{N,M}(t),\numop\Phi_{N,M}(t)\rangle + C_t M^3 N^{2\kappa - 3},
\end{align*}
where we chose $\eta= M^{-1} N^{3/2}$ in the last step. 
In conclusion, we have
\begin{multline}\label{eq:R3K:bound}
	 \pm\Im\langle\Phi_{N,M}(t),\chi_M[R_3,\K]\chi_M\Phi_{N,M}(t)\rangle \\
	 \leq C_t( \langle\Phi_{N,M}(t),\numop\Phi_{N,M}(t)\rangle + N^{2\kappa-1} + M^3N^{2\kappa -3} ). 
\end{multline}

\paragraph{Estimate of the $[R_3^\ast,\K]$ term.} 
On $\mathcal{F}_\perp^{\leq N}(t)$, we have in the sense of quadratic forms
\begin{align*}
	[R_3^\ast,\K] &= 2 \omega_3(\numop-1) \biggl(\iiint w_N(x-y)u_N(t,x) k_\perp(t,y,x') a^\ast_ya^\ast_{x}a_{x'}^\ast\diff x\diff y\diff x'\biggr) \\
	&\quad + (\omega_3(\numop-1)-\omega_3(\numop-3)) \K_\perp \biggl(\iint w_N(x-y)u_N(t,x)a^\ast_ya_x^\ast a_y\diff x\diff y\biggr). 
\end{align*}
Similarly to the step before, we obtain 
\begin{equation}\label{eq:R3astK:bound}
	 \pm\Im\langle\Phi_{N,M}(t),\chi_M[R_3^\ast,\K]\chi_M\Phi_{N,M}(t)\rangle \leq C_t( \langle\Phi_{N,M}(t),\numop\Phi_{N,M}(t)\rangle + MN^{2\kappa -1} ). 
\end{equation}

\paragraph{Estimate of the $[R_4,\K]$ term.} 
On $\mathcal{F}_\perp^{\leq N}(t)$, we have in the sense of quadratic forms
\begin{align*}
	[R_4,\K] &= \frac{1}{N-1} \biggl( \iint w_N(x-y) k_\perp(t,x,y) a^\ast_xa^\ast_y \diff x \diff y \\ 
	&\quad + 2 \iiint w_N(x-y)k_\perp(t,y,x') a^\ast_xa^\ast_ya^\ast_{x'}a_x \diff x\diff y\diff x' \biggr). 
\end{align*}
Similarly to the steps before, we obtain
\begin{multline}\label{eq:R4K:bound}
	 \pm\Im\langle\Phi_{N,M}(t),\chi_M[R_4,\K]\chi_M\Phi_{N,M}(t)\rangle \\
	  \leq C_t( \langle\Phi_{N,M}(t),\numop\Phi_{N,M}(t)\rangle + M^{1/2}N^{2\kappa-1} + M^2N^{2\kappa-2} ). 
\end{multline}

\paragraph{Estimate of the $\chi_M \calENt \K \delta_+ + \delta_- \K \calENt \chi_M$ term.}
We write 
\[
	\chi_M \calENt \K \delta_+ = \chi_M \K \calENt \delta_+ + \chi_M [\calENt, \K] \delta_+. 
\]
Using 
\[
	\pm\delta_+ \leq \frac{C}{M} \ind(M/2 \leq \numop \leq M),
\]
$\chi_M [\calENt, \K] \delta_+$ can be bounded similarly to the terms in the steps before and we obtain
\[
	  \pm\Im\langle\Phi_{N,M}(t),\chi_M [\calENt, \K] \delta_+\Phi_{N,M}(t)\rangle \leq C_t( \langle\Phi_{N,M}(t),\numop\Phi_{N,M}(t)\rangle + M^{-1}N^{2\kappa-1}). 
\]
The terms $\chi_M \K(R_j+R_j^\ast)\delta_+$, $j=0,\dotsc,4$, are all bounded similarly. We bound, for example, the terms containing $R_3$ and $R_4$, which contribute the largest errors. Using $\pm\omega_3(\numop)\leq CN^{-1/2}$, \eqref{eq:offd:bound}, and Lemmas~\ref{lem:MT} and \ref{lem:truncatedkinetic}, we have
\begin{align*}
	\MoveEqLeft[3]  \pm\Im\langle\Phi_{N,M}(t),\chi_M\K R_3 \delta_+\Phi_{N,M}(t)\rangle \\
	&= \pm \iint w_N(x-y)\overline{u_N(t,x)} \Im \langle a_y\omega_3(\numop)\chi_M\Phi_{N,M}(t), \K_\perp a_xa_y \delta_+ \Phi_{N,M}(t)\rangle \diff x\diff y \\
	&\leq C\norm{k(t)}_{L^2}\norm{u_N(t)}_{L^\infty} \Bigl( \eta \norm{w_N}_{L^1} MN^{-1} \langle \Phi_{N,M}(t), \numop \Phi_{N,M}(t) \rangle \\
	&\quad + \eta^{-1} N^{\kappa} \langle \Phi_{N,M}(t), \delta_+ (\numop+1)^2\sq(1-\Delta) \delta_+ \Phi_{N,M}(t) \rangle \Bigr) \\
	&\leq C_t ( \eta MN^{-1}  \langle \Phi_{N,M}(t), \numop \Phi_{N,M}(t) \rangle + \eta^{-1} N^{2\kappa} ) \\
	&\leq C_t ( \langle \Phi_{N,M}(t), \numop \Phi_{N,M}(t) \rangle + MN^{2\kappa-1} ),
\end{align*}
where we chose $\eta=M^{-1}N$.
Similarly, we obtain
\begin{align*}
	\MoveEqLeft[3]  \pm\Im\langle\Phi_{N,M}(t),\chi_M\K R_4 \delta_+\Phi_{N,M}(t)\rangle \\
	&= \pm \frac{1}{2(N-1)} \iint w_N(x-y) \Im \langle a_xa_y\chi_M\Phi_{N,M}(t), \K_\perp a_xa_y \delta_+ \Phi_{N,M}(t)\rangle \diff x\diff y \\
	&\leq \frac{CM}{N} \norm{k(t)}_{L^2} \iint\abs{w_N(x-y)} \Bigl( \eta \langle\Phi_{N,M}(t),\chi_Ma^\ast_xa^\ast_ya_xa_y\chi_M\Phi_{N,M}(t)\rangle \\
	&\quad + \eta^{-1}\langle\Phi_{N,M}(t),\delta_+a^\ast_xa^\ast_ya_xa_y\delta_+\Phi_{N,M}(t)\rangle\Bigr) \diff x \diff y \\
	&\leq \frac{C_tM}{N} ( \eta MN^{2\kappa} + \eta^{-1} M^{-1} N^{2\kappa}) \leq C_t MN^{2\kappa-1},
\end{align*}
where we chose $\eta=M^{-1}$. 
The $\delta_- \K \calENt \chi_M$ term can be estimated similarly to the $\chi_M \calENt \K \delta_+$ term and yields the same error bounds. In conclusion, we have
\begin{multline}\label{eq:delta:bound}
	 \pm\Im\langle\Phi_{N,M}(t),(\chi_M \calENt \K \delta_+ + \delta_- \K \calENt \chi_M)\Phi_{N,M}(t)\rangle \\
	 \leq C_t( \langle\Phi_{N,M}(t),\numop\Phi_{N,M}(t)\rangle + MN^{2\kappa-1}). 
\end{multline}
Collecting the bounds \eqref{eq:R0K:bound}, \eqref{eq:R1K:bound}, \eqref{eq:R2K:bound}, \eqref{eq:R3K:bound}, \eqref{eq:R3astK:bound}, \eqref{eq:R4K:bound}, and \eqref{eq:delta:bound}, we obtain
\[
	-I_1 \leq C_t\langle\Phi_{N,M}(t),(\numop+1)\Phi_{N,M}(t)\rangle + C_tMN^{2\kappa-1}.
\]
Thus, we have
\begin{multline*}
	\frac{\diff}{\diff t}\biggl(\langle \Phi_{N,M}(t),\numop\Phi_{N,M}(t)\rangle + 2\Re\langle\Phi_{N,M}(t),\iint k(t,x,y)a^\ast_xa^\ast_y \diff x\diff y \Phi_{N,M}(t)\rangle\biggr) \\
	\leq C_t\langle\Phi_{N,M}(t),(\numop+1)\Phi_{N,M}(t)\rangle + C_tMN^{2\kappa-1}. 
\end{multline*}
Using a Gr{\"o}nwall argument as in the proof of \eqref{eq:number:bog}, we get
\[
	\langle \Phi_{N,M}(t),\numop\Phi_{N,M}(t)\rangle \leq C_t(1+MN^{2\kappa-1}). 
\]
By our assumption $\delta>2\kappa$, we have $MN^{2\kappa-1}=N^{2\kappa-\delta}\leq 1$, which concludes the proof. 
\end{proof}

Now we can compare $\Phi_N$ with the truncated dynamics. 

\begin{lem}\label{lem:tt}
Let $M_1=N^{1-\delta_1}$, $M_2=N^{1-\delta_2}$, $\delta_1,\delta_2\in(0,1)$, $\delta_1<\delta_2$.
If $w_N$ satisfies \eqref{eq:poly} for $\beta>0$, we have for $N$ large enough
\begin{equation}\label{eq:poly:f1}
	\norm{\Phi_N(t)-\Phi_{N,M_1}(t)}^2 \leq \frac{C_{t}\beta\log(N)N^{1/2}}{M_1}
\end{equation}
and 
\begin{equation}\label{eq:poly:12}
	\norm{\Phi_{N,M_1}(t)-\Phi_{N,M_2}(t)}^2 \leq \frac{C_{t}(\beta\log(N))^{3/2}}{M_2}. 
\end{equation}
If $w_N$ satisfies \eqref{eq:exp} for $0<\kappa<1$, we require additionally $\delta_1,\delta_2\in(\kappa,1)$. 
Then, we have for $N$ large enough
\begin{equation}\label{eq:exp:f1}
	\norm{\Phi_N(t)-\Phi_{N,M_1}(t)}^2 \leq \frac{C_{t}N^{\kappa+1/2}}{M_1}
\end{equation}
and
\begin{equation}\label{eq:exp:12}
	\norm{\Phi_{N,M_1}(t)-\Phi_{N,M_2}(t)}^2 \leq \frac{C_{t}N^{3\kappa/2}}{M_2}. 
\end{equation}
\end{lem}

Let $f\colon\R\to[0,1]$ be a smooth function such that $f(r)=1$ if $r\leq1/4$ and $f(r)=0$ if $r\geq1/2$, and define the operator $f_{M}$ on $\mathcal{F}$ by
\begin{equation}\label{eq:smoothcut}
	f_M = f\biggl(\frac{\numop}{M}\biggr). 
\end{equation}
Note that $f_M^2\chi_M = f_M^2$ and $\ind^{\leq M/2 + 2}\chi_M = \ind^{\leq M/2 + 2}$ for $M\geq8$ and thus, 
\begin{align*}
	f_M^2\mathcal{G}_{N,M}(t) &= f_M^2\BH(t) + f_M^2\chi_M\mathcal{E}_N(t)\chi_M = f_M^2\BH(t) + f_M^2\mathcal{E}_N(t)\ind^{\leq M/2 + 2}\chi_M \\
	&= f_M^2\BH(t) + f_M^2\mathcal{E}_N(t) \ind^{\leq M/2 + 2} = f_M^2\BH(t) + f_M^2\mathcal{E}_N(t) = f_M^2\mathcal{G}_N(t). 
\end{align*}
To prove Lemma~\ref{lem:tt}, we will need the following result. 

\begin{lem}\label{lem:average}
Let $M=N^{1-\delta}$, $\delta\in(0,1)$. If $w_N$ satisfies \eqref{eq:poly} for $\beta>0$, we have for $\xi',\xi\in\mathcal{F}_\perp(t)$
\begin{equation*}
	\abs{\langle\xi',[\mathcal{G}_N(t),f_{M}^2]\xi\rangle} \leq \frac{C_{t}(\beta\log(N))^{1/2}}{M}\norm{\sq(1-\Delta)^{1/2}\xi'}\norm{\sq(1-\Delta)^{1/2}\xi}. 
\end{equation*}
If $w_N$ satisfies \eqref{eq:exp} for $0<\kappa<1$, $\beta\log(N)$ is replaced by $N^{\kappa}$. 
\end{lem}

\begin{proof}
We focus on the case where $w_N$ satisfies \eqref{eq:poly}, since the case where $w_N$ satisfies \eqref{eq:exp} follows by setting $\beta=N^\kappa/\log(N)$. 
Let $\xi',\xi\in\mathcal{F}_\perp(t)$. We start by writing 
\[
	[\mathcal{G}_N(t),f_{M}^2] = \frac12\Bigl([R_1+R_1^\ast,f_{M}^2] + [G_2+G_2^\ast,f_{M}^2] + [R_3+R_3^\ast,f_{M}^2]\Bigr),
\]
where 
\begin{align*}
	G_2 &= \iint K_2(t,x,y)a^\ast_xa^\ast_y\diff x\diff y + R_2 \\
	&= \iint K_2(t,x,y) a^\ast_x a^\ast_y \diff x \diff y \biggl(\frac{\sqrt{(N-\numop)(N-\numop-1)}}{N-1}\biggr). 
\end{align*}
We estimate all three terms separately. 

\paragraph{Estimate of the linear terms.}

Recall 
\[
	R_1=-2\frac{\numop\sqrt{N-\numop}}{N-1} a(q(t)[w_N\ast\abs{u_N(t)}^2]u_N(t)). 
\]
We calculate 
\[
	[R_1,f_{M}^2] = \tilde{\omega}_1 a(q(t)[w_N\ast\abs{u_N(t)}^2]u_N(t)),
\]
where
\[
	\tilde{\omega}_1 = -2\frac{\numop\sqrt{N-\numop}}{N-1} \biggl(f^2\biggl(\frac{\numop+1}{M}\biggr) - f^2\biggl(\frac{\numop}{M}\biggr)\biggr). 
\]
We observe that for $N\geq2$ 
\[
	\abs{\tilde{\omega}_1} \leq \frac{C\numop}{M\sqrt{N}}\ind^{\leq M}. 
\]
Using the simple estimate
\[
	a^\ast(W_1)a(W_1) \leq \norm{W_1}_{L^2}^2\numop,
\]
where $W_1=q(t)[w_N\ast\abs{u_N(t)}^2]u_N(t)$ satisfies
\[
	\norm{W_1}_{L^2} \leq \norm{w_N}_{L^1} \norm{u_N(t)}_{L^\infty}^2 \norm{u_N(t)}_{L^2} \leq C_t,
\]
we obtain
\begin{align*}
\abs{\langle\xi',[R_1,f_{M}^2]\xi\rangle} &\leq \langle\xi',\abs{\tilde{\omega}_1}\xi'\rangle^{1/2} \langle\xi,a^\ast(W_1)\abs{\tilde{\omega}_1}a(W_1)\xi\rangle^{1/2} \\
&\leq \frac{C}{M\sqrt{N}} \langle\xi',\numop\xi'\rangle^{1/2} \langle\xi,a^\ast(W_1)\numop\ind^{\leq M}a(W_1)\xi\rangle^{1/2} \\
&\leq \frac{C_t}{\sqrt{MN}} \langle\xi',\numop\xi'\rangle^{1/2} \langle\xi,\numop\xi\rangle^{1/2}. 
\end{align*}
Similarly, we get 
\[
	\abs{\langle\xi',[R_1^\ast,f_{M}^2]\xi\rangle} \leq \frac{C_t}{\sqrt{MN}} \langle\xi',\numop\xi'\rangle^{1/2} \langle\xi,\numop\xi\rangle^{1/2}. 
\]

\paragraph{Estimate of the quadratic terms.}

We calculate 
\[
	[G_2,f_{M}^2] = \biggl(\iint K_2(t,x,y) a^\ast_x a^\ast_y \diff x \diff y\biggr) \tilde{\omega}_2,
\]
where
\[
	\tilde{\omega}_2 = \frac{\sqrt{(N-\numop)(N-\numop-1)}}{N-1} \biggl(f^2\biggl(\frac{\numop}{M}\biggr) - f^2\biggl(\frac{\numop+2}{M}\biggr)\biggr).
\]
For $N\geq2$, we have $\tilde{\omega}_2 = \tilde{\omega}_2\ind^{\leq M}$ and 
\[
	\abs{\tilde{\omega}_2}^2 \leq \frac{C}{M^2}\ind\biggl(\frac{M}{4}-2 < \numop \leq \frac{M}{2}\biggr)\leq \frac{C\numop}{M^3}. 
\]
Using the bound
\begin{equation}\label{eq:absR4bound}
	\ind^{\leq M}\sq_2(\abs{w_N})\ind^{\leq M}\leq C\beta\log(N) M\sq(1-\Delta),
\end{equation}
which follows from Lemma~\ref{lem:MT}, we obtain
\begin{align*}
\abs{\langle\xi',[G_2,f_{M}^2]\xi\rangle} 
	&\leq \iint \langle\ind^{\leq M}\xi',a_x^\ast a_y^\ast\abs{w_N(x-y)}a_ya_x\ind^{\leq M}\xi'\rangle^{1/2}\times \\
	&\quad\times \langle\xi,\abs{u_N(t,x)}^2\abs{w_N(x-y)}\abs{u_N(t,y)}^2\abs{\tilde{\omega}_2}^2\xi\rangle^{1/2} \diff x\diff y \\
	&\leq \eta \langle\ind^{\leq M}\xi',\sq_2(\abs{w_N})\ind^{\leq M}\xi'\rangle \\
	&\quad + \eta^{-1}C\norm{u_N(t)}_{L^\infty}^2\norm{w_N}_{L^1}\norm{u_N(t)}_{L^2}^2\langle\xi,\abs{\tilde{\omega}_2}^2\xi\rangle \\
	&\leq \eta C\beta\log(N) M\langle\xi',\sq(1-\Delta)\xi'\rangle + \eta^{-1} C_{t}M^{-3} \langle\xi,\numop\xi\rangle \\
	&\leq \frac{C_{t} (\beta\log(N))^{1/2}}{M} \langle\xi',\sq(1-\Delta)\xi'\rangle^{1/2} \langle\xi,\numop\xi\rangle^{1/2},
\end{align*}
where we optimized over $\eta>0$ in the last step. Here, we could omit the projection $q(t)$ since $\xi',\xi\in\mathcal{F}_\perp(t)$. 
In the same way, we get 
\[
	\abs{\langle\xi',[G_2^\ast,f_{M}^2]\xi\rangle} \leq \frac{C_{t} (\beta\log(N))^{1/2}}{M} \langle\xi',\numop\xi'\rangle^{1/2} \langle\xi,\sq(1-\Delta)\xi\rangle^{1/2}. 
\]

\paragraph{Estimate of the cubic terms.}

Recall
\[
R_{3}=2\frac{\sqrt{N-\numop}}{N-1}\iiiint( \id \otimes q(t) w_N q(t)\otimes q(t))(x,y;x',y')\overline{u_N(t,x)} a^\ast_y a_{x'} a_{y'} \diff x \diff y \diff x' \diff y'. 
\]
We calculate
\[
	[R_3,f_{M}^2] = \tilde{\omega}_3 \biggl(\iiiint( \id \otimes q(t) w_N q(t)\otimes q(t))(x,y;x',y')\overline{u_N(t,x)} a^\ast_y a_{x'} a_{y'} \diff x \diff y \diff x' \diff y'\biggr),
\]
where
\[
	\tilde{\omega}_3 = 2\frac{\sqrt{N-\numop}}{N-1} \biggl(f^2\biggl(\frac{\numop+1}{M}\biggr) - f^2\biggl(\frac{\numop}{M}\biggr)\biggr). 
\]
For $N\geq 2$, we have $\tilde{\omega}_3=\tilde{\omega}_3\ind^{\leq M}$ and 
\[
	\abs{\tilde{\omega}_3}^2 \leq \frac{C}{M^2N}\ind^{\leq M}. 
\]
Since $\xi',\xi\in\mathcal{F}_\perp(t)$, we can omit the projection $q(t)$ and write
\[
	\langle\xi',[R_3,f_{M}^2]\xi\rangle = \iint\langle\xi',\tilde{\omega}_3w_N(x-y)\overline{u_N(t,x)} a^\ast_ya_ya_x\xi\rangle \diff x\diff y. 
\]
Using \eqref{eq:absR4bound} again, we obtain
\begin{align*}
	\abs{\langle\xi',[R_3,f_{M}^2]\xi\rangle} 
	&\leq \iint \langle\xi',\abs{w_N(x-y)}\abs{u_N(t,x)}^2\abs{\tilde{\omega}_3}a^\ast_ya_y\abs{\tilde{\omega}_3}\xi'\rangle^{1/2}\times \\
	&\quad\times \langle\ind^{\leq M}\xi,a^\ast_xa^\ast_y\abs{w_N(x-y)}a_ya_x\ind^{\leq M}\xi\rangle^{1/2} \diff x\diff y\\
	&\leq \frac{C\eta}{M^2N} \norm{w_N}_{L^1}\norm{u_N(t)}_{L^\infty}^2\langle\xi',\numop\xi'\rangle \\
	&\quad + \eta^{-1}\langle\ind^{\leq M}\xi,\sq_2(\abs{w_N})\ind^{\leq M}\xi\rangle \\
	&\leq \frac{C_{t}\eta}{M^2N}\langle\xi',\numop\xi'\rangle + \eta^{-1}C\beta\log(N) M\langle\xi,\sq(1-\Delta)\xi\rangle \\
	&\leq \frac{C_{t}(\beta\log(N))^{1/2}}{\sqrt{MN}}\langle\xi',\numop\xi'\rangle^{1/2}\langle\xi,\sq(1-\Delta)\xi\rangle^{1/2},
\end{align*}
where we optimized over $\eta>0$ in the last step. Similarly, we get
\[
	\abs{\langle\xi',[R_3^\ast,f_{M}^2]\xi\rangle} \leq \frac{C_{t}(\beta\log(N))^{1/2}}{\sqrt{MN}}\langle\xi',\sq(1-\Delta)\xi'\rangle^{1/2}\langle\xi,\numop\xi\rangle^{1/2}. \qedhere
\]
\end{proof}

\begin{proof}[Proof of Lemma~\ref{lem:tt}]
The proof is similar to that of \cite[Lemma~12]{NamNap-19} and \cite[Proposition~9]{BosDieNam-26}, and we only sketch the main steps if $w_N$ satisfies \eqref{eq:exp} for $0<\kappa<1$. 

To prove \eqref{eq:exp:f1}, we use
\[
	\norm{\Phi_N(t)-\Phi_{N,M_1}(t)}^2 \leq 2\Bigl(1-\Re\langle\Phi_N(t),\Phi_{N,M_1}(t)\rangle\Bigr). 
\]
Define $f_{M_1}$ as in \eqref{eq:smoothcut} and similarly the operator $g_{M_1}$, where $g^2=1-f^2$. We write
\[
	\langle\Phi_N(t),\Phi_{N,M_1}(t)\rangle = \langle\Phi_N(t),f_{M_1}^2\Phi_{N,M_1}(t)\rangle + \langle\Phi_N(t),g_{M_1}^2\Phi_{N,M_1}(t)\rangle.
\]
Using the bound
\[
	g_{M_1}^2 \leq \ind^{>M_1/4} \leq \frac{4\numop}{M_1}
\] 
and Lemma~\ref{lem:truncatedkinetic}, we estimate
\begin{equation}\label{eq:mpsbound}
\begin{aligned}
	\abs{\langle\Phi_N(t),g_{M_1}^2\Phi_{N,M_1}(t)\rangle} & \leq \norm{g_{M_1}\Phi_N(t)}\cdot\norm{g_{M_1}\Phi_{N,M_1}(t)} \\
	&\leq C \langle \Phi_{N,M_1}(t),(\numop/M_1)\Phi_{N,M_1}(t)\rangle^{1/2} \leq C_{t}N^{\kappa/2}M_1^{-1/2}. 
\end{aligned}
\end{equation}
For the other term, we compute
\[
	\frac{\diff}{\diff t}\langle\Phi_N(t),f_{M_1}^2\Phi_{N,M_1}(t)\rangle = \langle\Phi_N(t),\ii[\mathcal{G}_N(t),f_{M_1}^2]\Phi_{N,M_1}(t)\rangle. 
\]
To estimate this term, we will use Lemma~\ref{lem:average}. Together with Lemmas~\ref{lem:apkinetic} and \ref{lem:truncatedkinetic} and the fact that $v\geq0$, we obtain
\begin{align*}
\MoveEqLeft[3] \abs[\bigg]{\frac{\diff}{\diff t}\langle\Phi_N(t),f_{M_1}^2\Phi_{N,M_1}(t)\rangle} = \abs[\big]{\langle\Phi_N(t),\ii[\mathcal{G}_N(t),f_{M_1}^2]\Phi_{N,M_1}(t)\rangle} \\
&\leq C_{t}\frac{N^{\kappa/2}}{M_1} \langle\Phi_N(t),\sq(1-\Delta)\Phi_N(t)\rangle^{1/2}\langle\Phi_{N,M_1}(t),\sq(1-\Delta)\Phi_{N,M_1}(t)\rangle^{1/2} \\
&\leq C_{t}\frac{N^{\kappa + 1/2}}{M_1}. 
\end{align*}
Therefore, we have
\[
	\Re\langle\Phi_N(t),f_{M_1}^2\Phi_{N,M_1}(t)\rangle
	\geq \Re\langle\Phi_N(0),f_{M_1}^2\Phi_{N,M_1}(0)\rangle - C_t\frac{N^{\kappa + 1/2}}{M_1}.  
\]
Since $\Phi_N(0)=\ind^{\leq N}\Phi_0$, $\Phi_{N,M_1}(0)=\Phi_0$, and $\langle\Phi_0,\numop\Phi_0\rangle\leq C$ by \eqref{eq:thm:norm:initial}, we have
\begin{align*}
	\langle\Phi_N(0),f_{M_1}^2\Phi_{N,M_1}(0)\rangle &= \langle\Phi_0,f_{M_1}^2\Phi_0\rangle = 1-\langle\Phi_0,g_{M_1}^2\Phi_0\rangle \\
	&\geq 1-C\langle\Phi_0,(\numop/M_1)\Phi_0\rangle \geq 1-CM_1^{-1}. 
\end{align*}
Thus, we obtain 
\[
	\Re\langle\Phi_N(t),f_{M_1}^2\Phi_{N,M_1}(t)\rangle
	\geq 1-C_{t}\frac{N^{\kappa + 1/2}}{M_1}. 
\]
Together with \eqref{eq:mpsbound}, we have
\[
	\Re\langle\Phi_N(t),\Phi_{N,M_1}(t)\rangle \geq 1-C_{t}\frac{N^{\kappa + 1/2}}{M_1}-C_{t}\sqrt{\frac{N^{\kappa}}{M_1}}. 
\]
Hence, after simplifying the bound, we get
\[
	\norm{\Phi_N(t)-\Phi_{N,M_1}(t)}^2 \leq 2\Bigl(1-\Re\langle\Phi_N(t),\Phi_{N,M_1}(t)\rangle\Bigr) \leq C_{t}\frac{N^{\kappa + 1/2}}{M_1},
\]
which concludes the proof of \eqref{eq:exp:f1}. 

The above proof can easily be adjusted to show \eqref{eq:exp:12}. 
\end{proof}

Finally, we compare the truncated dynamics $\Phi_{N,M}$ with the Bogoliubov dynamics $\Phi$. 

\begin{lem}\label{lem:tbog}
Let $M=N^{1-\delta}$, $\delta\in(0,1)$. If $w_N$ satisfies \eqref{eq:poly} for $\beta>0$, we have for $N$ large enough
\begin{equation}\label{eq:poly:2bog}
	\norm{\Phi_{N,M}(t)-\Phi(t)}^2 \leq C_{t}\beta\log(N)\sqrt{\frac{M}{N}}.  
\end{equation}
If $w_N$ satisfies \eqref{eq:exp} for $0<\kappa<1/2$, we require additionally that $\delta\in(2\kappa,1)$. Then, we have for $N$ large enough
\begin{equation}\label{eq:exp:2bog}
	\norm{\Phi_{N,M}(t)-\Phi(t)}^2 \leq C_{t}N^{\kappa}\sqrt{\frac{M}{N}}. 
\end{equation}
\end{lem}

\begin{proof}
The proof is similar to that of Lemma~\ref{lem:tt}, and we only sketch the main steps if $w_N$ satisfies \eqref{eq:exp} for $0<\kappa<1/2$. 

To prove \eqref{eq:exp:2bog}, we use
\[
	\norm{\Phi_{N,M}(t)-\Phi(t)}^2 \leq 2\Bigl(1-\Re\langle\Phi_{N,M}(t),\Phi(t)\rangle\Bigr). 
\]
We compute
\[
	\frac{\diff}{\diff t}\langle\Phi_{N,M}(t),\Phi(t)\rangle = \ii\langle\Phi_{N,M}(t),\chi_M\mathcal{E}_N(t)\chi_M\Phi(t)\rangle. 
\]
Using Lemmas~\ref{lem:error-bounds}, \ref{lem:bogdyn}, \ref{lem:truncatedkinetic}, and \ref{lem:truncatednumber}, we have
\begin{align*}
	\MoveEqLeft[3] \abs{\langle\Phi_{N,M}(t),\chi_M\mathcal{E}_N(t)\chi_M\Phi(t)\rangle} \\
	&\leq \frac{CN^\kappa M}{N}\langle\Phi_{N,M}(t),\sq(1-\Delta)\Phi_{N,M}(t)\rangle^{1/2}\langle\Phi(t),\sq(1-\Delta)\Phi(t)\rangle^{1/2} \\
	&\quad + C_t\sqrt{\frac{N^\kappa M}{N}} \Bigl( \langle\Phi_{N,M}(t),\sq(1-\Delta)\Phi_{N,M}(t)\rangle^{1/2}\langle\Phi(t),\numop\Phi(t)\rangle^{1/2} \\
	&\quad + \langle\Phi_{N,M}(t),\numop\Phi_{N,M}(t)\rangle^{1/2}\langle\Phi(t),\sq(1-\Delta)\Phi(t)\rangle^{1/2} \Bigr) \\
	&\leq C_{t}\biggl(\frac{N^{2\kappa}M}{N} + N^{\kappa}\sqrt{\frac{M}{N}}\biggr)
	\leq C_tN^{\kappa}\sqrt{\frac{M}{N}}. 
\end{align*}
Therefore, we obtain
\[
	\norm{\Phi_{N,M}(t)-\Phi(t)}^2 \leq 2\biggl(1 - \Re\langle\Phi_{N,M}(0),\Phi(0)\rangle + C_tN^{\kappa}\sqrt{\frac{M}{N}} \biggr) 
	\leq C_{t}N^{\kappa}\sqrt{\frac{M}{N}},
\]
which concludes the proof of \eqref{eq:exp:2bog}. 
\end{proof}

\section{Conclusion of the proof of the main results}\label{sec:con}

\begin{proof}[Proof of Theorem~\ref{thm:norm}]
Let $M_1=N^{1-\delta_1}$ and $M_2=N^{1-\delta_2}$ with $\delta_1,\delta_2\in(0,1)$, $\delta_1<\delta_2$. 
Recall the definitions \eqref{eq:trafodyn}, \eqref{eq:truncdyn}, and \eqref{eq:Bogoliubov} of $\Phi_N(t)$, $\Phi_{N,M}(t)$, and $\Phi(t)$, respectively. To prove \eqref{eq:normapprox}, it is enough to show
\[
	\norm{\Phi_N(t)-\Phi(t)}^2 \leq C_tN^{-\alpha}
\]
since $U_N(t)\colon\mathfrak{H}^N\to\mathcal{F}_\perp^{\leq N}(t)$ is unitary and we thus have
\[
	\norm{\Psi_N(t)-U_N(t)^\ast\ind^{\leq N}\Phi(t)} = \norm{U_N(t)\Psi_N(t)-\ind^{\leq N}\Phi(t)} \leq \norm{\Phi_N(t)-\Phi(t)}. 
\]
We first treat the case where $w_N$ satisfies \eqref{eq:poly} with $\beta>0$. From Lemmas~\ref{lem:tt} and \ref{lem:tbog}, we obtain
\begin{align*}
	\MoveEqLeft[3] \norm{\Phi_N(t)-\Phi(t)}^2 \\
	&\leq 3\norm{\Phi_N(t)-\Phi_{N,M_1}(t)}^2 + 3\norm{\Phi_{N,M_1}(t)-\Phi_{N,M_2}(t)}^2 + 3\norm{\Phi_{N,M_2}(t)-\Phi(t)}^2 \\
	&\leq C_{t}(\beta\log(N))^{3/2}\biggl(\frac{N^{1/2}}{M_1} + \frac{1}{M_2} + \sqrt{\frac{M_2}{N}}\biggr) \\
	&= C_{t}(\beta\log(N))^{3/2}\bigl( N^{-1/2+\delta_1} + N^{-1+\delta_2} + N^{-\delta_2/2} \bigr)
\end{align*}
for all $t\geq0$. Choosing $\delta_1=1/6$ and $\delta_2=2/3$ shows \eqref{eq:normapprox} for any $0<\alpha<1/3$. 

To prove the case where $w_N$ satisfies \eqref{eq:exp} with $0<\kappa<1/4$, we require additionally $\delta_1>\kappa$ and $\delta_2>2\kappa$. From Lemmas~\ref{lem:tt} and \ref{lem:tbog}, we obtain
\begin{align*}
\norm{\Phi_N(t)-\Phi(t)}^2
	&\leq C_{t}\biggl(\frac{N^{\kappa+1/2}}{M_1} + \frac{N^{3\kappa/2}}{M_2} + N^{\kappa}\sqrt{\frac{M_2}{N}}\biggr) \\
	&= C_{t}\bigl( N^{\kappa-1/2+\delta_1} + N^{3\kappa/2-1+\delta_2} + N^{\kappa-\delta_2/2}\bigr)
\end{align*}
for all $t\geq0$. We choose $\delta_2=(2-\kappa)/3>2\kappa$ and $\delta_1=\kappa+\varepsilon$ with $0<\varepsilon<\delta_2-\kappa$, which concludes the proof of \eqref{eq:normapprox} for $0<\alpha<\min(1/2-2\kappa,1/3-7\kappa/6)$.
\end{proof}

\begin{proof}[Proof of Corollary~\ref{cor:condensate}]
The fact that the norm convergence implies the convergence of density matrices is well known; see e.g.\ \cite[Corollary~2]{LewNamSch-15}. 
Hence, we infer from Theorem~\ref{thm:norm} 
\[
	\lim_{N\to\infty} \Tr\abs{\gamma_{\Psi_N(t)}^{(1)} - \ketbra{u_N(t)}{u_N(t)}} = 0. 
\]
It remains to show 
\[
	\lim_{N\to\infty} \Tr\abs{\ketbra{u_N(t)}{u_N(t)}-\ketbra{\varphi(t)}{\varphi(t)}} = 0,
\]
which follows from $\norm{u_N(t)-\varphi(t)}_{L^2}\to0$ since
\[
	\Tr\abs{\ketbra{u_N(t)}{u_N(t)}-\ketbra{\varphi(t)}{\varphi(t)}} \leq 2\norm{u_N(t)-\varphi(t)}_{L^2}. 
\]
This can easily be proved using the fact that $w_N\ast\abs{\varphi(t)}^2\to a\abs{\varphi(t)}^2$ in $L^2$. 
\end{proof}

\begin{rem}
If $\widehat{w}$ is Lipschitz, $\norm{u_N(t)-\varphi(t)}_{L^2}$ is of order $\mathcal{O}(N^{-\beta})$ if $w_N$ satisfies \eqref{eq:poly} and of order $\mathcal{O}(\ee^{-N^\kappa})$ if $w_N$ satisfies \eqref{eq:exp}. The proof of \cite[Lemma~10]{BosDieNam-26} shows the case $V\equiv0$ and can easily be adapted to the case of some external potential $V$ that satisfies Assumption~\ref{assu:V}. 
\end{rem}

\subsection*{Acknowledgements}
LR thanks Lea Bo{\ss}mann for valuable feedback and discussions. LR also thanks Fran\c{c}ois Visconti for drawing his attention to \cite{JunVis-26}. 

\subsection*{AI declaration}
The main result was proved by the author without the aid of AI tools for all $\beta>0$ and $0<\kappa<2/9$ using a truncated dynamics as introduced in \cite{NamNap-19}. Then, AI tools were used to improve the range of $\kappa$ to $0<\kappa<1/4$: Opus 5.5 suggested using a different truncated dynamics (defined in \eqref{eq:truncdyn}) compared to \cite{NamNap-19} and assisted in the proof of the bounds on the number of particles of the truncated and Bogoliubov dynamics in Lemma~\ref{lem:truncatednumber} and Lemma~\ref{lem:bogdyn}~\ref{item:number:bog}, respectively. This improved Lemma~\ref{lem:tbog} and shortened its proof. AI tools were also used for proofreading. All proofs were written and checked by the author, who takes full responsibility for the results. 

\emergencystretch=1em
\begin{refcontext}[sorting=nyt]
\printbibliography[title={References},heading=bibintoc]
\end{refcontext}

\end{document}